\documentclass[envcountsame,numbook,smallextended]{svjour3}
\pdfoutput=1
\smartqed
\usepackage{amssymb,amsmath}
\usepackage{txfonts}
\usepackage[utf8]{inputenc}
\usepackage[T1]{fontenc}
\usepackage[english]{babel}
\usepackage{hyperref}
\hypersetup{colorlinks,allcolors=black}
\usepackage{microtype}
\usepackage{booktabs}
\usepackage[numbers,sort&compress]{natbib}
\usepackage{mathtools}
\usepackage[normalem]{ulem}
\usepackage{tikz}
\newcommand{\shorten}[1]{}
\newcounter{mycounter}

\pgfdeclarelayer{background}
\pgfsetlayers{background,main}

\usepackage{xcolor}
\usepackage{paralist}
\usepackage[textsize=scriptsize]{todonotes}
\numberwithin{equation}{section}
\numberwithin{figure}{section}
\usepackage[sort&compress,noabbrev,capitalize]{cleveref}

\spnewtheorem{observation}[theorem]{Observation}{\bfseries}{\rmfamily}
\spnewtheorem{rrule}[theorem]{Reduction Rule}{\bfseries}{\rmfamily}
\spnewtheorem{construction}[theorem]{Construction}{\bfseries}{\rmfamily}
\Crefname{construction}{Construction}{Constructions}

\makeatletter\let\c@lemma\relax\makeatother
\spnewtheorem{lemma}[theorem]{Lemma}{\bfseries}{\rmfamily}
\Crefname{observation}{Observation}{Observations}
\Crefname{problem}{Problem}{Problems}
\Crefname{lemma}{Lemma}{Lemmas}
\Crefname{section}{Section}{Sections}
\Crefname{rrule}{Reduction Rule}{Reduction Rules}
\Crefname{figure}{Figure}{Figures}

\author{René van Bevern\thanks{René van Bevern is supported by grant 16-31-60007 mol\textunderscore{}a\textunderscore{}dk of the Russian Foundation for Basic Research.} \and 
Vincent Froese\and Christian~Komusiewicz\thanks{Christian Komusiewicz is supported by grant KO~3669/4-1 of Deutsche Forschungsgemeinschaft.}}

\date{\today{}}

\journalname{Theory of Computing Systems}

\institute{René van Bevern\at Novosibirsk State University, Novosibirsk, Russian Federation, \email{rvb@nsu.ru}
\at
Sobolev Institute of Mathematics, Siberian Branch of the Russian Academy of Sciences, Novosibirsk, Russian Federation
\and Vincent Froese\at Technische Universität Berlin, Germany, \email{vincent.froese@tu-berlin.de}
\and Christian Komusiewicz\at Friedrich-Schiller-Universität Jena, Germany, \email{christian.komusiewicz@uni-jena.de}}

\crefname{property}{Property}{Properties}
\creflabelformat{property}{(#2#1#3)}
\crefname{corollary}{Corollary}{Corollaries}
\crefname{problem}{Problem}{Problems}
\newcommand{\decprob}[3]{%
\pagebreak[3]
  \begin{problem}[\boldmath#1]
    \begin{compactdesc}
        \item[\normalfont\it Input:] #2
        \item[\normalfont\it Question:] #3
    \end{compactdesc}
  \end{problem}
}%

\newcommand{\occ}{\ensuremath{\alpha}}
\newcommand{\num}{\ensuremath{\gamma}}
\newcommand{\symdiff}{\triangle}
\newcommand{\true}{\ensuremath{\mathrm{true}}}
\newcommand{\false}{\ensuremath{\mathrm{false}}}
\newcommand{\ffed}{\textsc{$F$-free Editing}}
\newcommand{\ffvd}{\textsc{$F$-free Vertex Deletion}}
\newcommand{\ffedv}{\textsc{$F$-free Editing with $F$-Packing}}
\newcommand{\tffedv}{\textsc{$F$-free Editing with Cost-$t$ Packing}}

\newcommand{\gfedv}[1]{\textsc{\ensuremath{#1}-free Editing with \ensuremath{#1}-Packing}}
\newcommand{\gfdev}[1]{\textsc{\ensuremath{#1}-free Deletion with \ensuremath{#1}-Packing}}

\newcommand{\gfedset}[1]{\ensuremath{#1}-free editing set}
\newcommand{\gfdeset}[1]{\ensuremath{#1}-free deletion set}
\newcommand{\packing}{\ensuremath{\mathcal H}}
\newcommand{\TAGP}{\textsc{Triangle Deletion with  Cost-$t$ Packing}}
\newcommand{\TATP}{\textsc{Triangle Deletion with  Triangle Packing}}

\newcommand{\CEGP}{\textsc{Cluster Editing with  Cost-$t$ Packing}}
\newcommand{\CEPP}{\textsc{Cluster Editing with  $P_3$-Packing}}
\newcommand{\POITS}{\textsc{3-SAT}}

\newcommand{\PLVD}{\textsc{\ensuremath{P_q}-free Vertex Deletion with  \ensuremath{P_q}-Packing}}
\newcommand{\FAST}{\textsc{Feedback Arc Set in Tournaments}}

\newcommand{\FASTGP}{\textsc{Feedback Arc Set in Tournaments with  Cost-$t$ Packing}}

\newcommand{\Time}{\ensuremath{\Gamma}}
\newcommand{\Wext}{\ensuremath{W_\text{ext}}}
\newcommand{\Wint}{\ensuremath{W_\text{int}}}
\title{Parameterizing edge modification problems\\above lower bounds\thanks{An extended abstract of this article appeared in Proceedings of the 11th International Computer Science Symposium in Russia, June 9–13, 2016, St.~Petersburg, Russian Federation~\citep{BFK16}.}}

\begin{document}

\maketitle
\pagestyle{plain}
\begin{abstract}
We study the parameterized complexity of a variant of the \ffed{} problem: Given a graph~$G$ and a natural number~$k$, is it possible to modify at most $k$~edges in~$G$ so that the resulting graph contains no induced subgraph isomorphic to~$F$?  In our variant, the input additionally contains a vertex-disjoint packing~$\packing$ of induced subgraphs of~$G$, which provides a lower bound~$h(\packing)$ on the number of edge modifications required to transform~$G$ into an $F$-free graph.  While earlier works used the number~$k$ as parameter or structural parameters of the input graph~$G$, we consider instead the parameter~$\ell:=k-h(\packing)$, that is, the number of edge modifications above the lower bound~$h(\packing)$.  We develop a framework of generic data reduction rules to show fixed-parameter tractability with respect to~$\ell$ for~\textsc{$K_3$-Free Editing}, \textsc{Feedback Arc Set in Tournaments}, and \textsc{Cluster Editing} when the packing~$\packing$ contains subgraphs with bounded solution size.  For~\textsc{$K_3$-Free Editing}, we also prove NP-hardness in case of edge-disjoint packings of~$K_3$s and~$\ell=0$, while for \textsc{$K_q$-Free Editing} and $q\ge 6$, NP-hardness for~$\ell=0$ even holds for vertex-disjoint packings of~$K_q$s.
  In addition, we provide NP-hardness results for \ffvd{}, were the
  aim is to delete a minimum number of vertices to make the input
  graph $F$-free.
\end{abstract}

\section{Introduction}
Graph modification problems are a core topic
of algorithmic research~\cite{LY80,Cai96,Yan81}.
Given a graph~$G$,
the aim is to transform~$G$ by a minimum number of modifications
(like vertex deletions, edge deletions, or edge insertions)
into another graph~$G'$ fulfilling certain properties.
Particularly well-studied are \emph{hereditary} graph properties, which are closed under vertex deletions and are characterized by \emph{minimal forbidden induced subgraphs}: a graph fulfills such a property if and only if it does not contain a graph~$F$ from a property-specific family~$\mathcal{F}$ of graphs as induced subgraph.  All nontrivial vertex deletion problems and many edge modification and deletion problems for establishing hereditary graph properties are NP-complete
\cite{LY80,Alon06,MK86,Yan81,ASS16}.  One approach to cope with the NP-hardness of these problems are \emph{fixed-parameter algorithms} that solve them in \(f(k)\cdot n^{O(1)}\)~time for some exponential function~\(f\) depending only on some desirably small parameter~\(k\).   If the desired graph property has a finite forbidden induced subgraph characterization, then the corresponding vertex deletion, edge deletion, and edge modification problems are \emph{fixed-parameter tractable} parameterized by the number of modifications~$k$, that is, solvable in \(f(k)\cdot n^{O(1)}\)~time~\cite{Cai96}.

\paragraph{Parameterization above lower bounds.} When combined with data reduction and pruning rules, search-tree based fixed-parameter algorithms for the parameter~\(k\) of allowed modifications can yield competitive problem solvers~\cite{HH15,MNS12}.  Nevertheless, the number of modifications is often too large and smaller parameters are desirable.

A natural approach to obtain smaller parameters is ``parameterization above guaranteed values''~\cite{MR99,CPPW13,LNR+14,GP15}. %
The idea is to use a lower bound~$h$ on the solution size and to use~$\ell:=k-h$ as parameter instead of~$k$.
This idea has been applied successfully to \textsc{Vertex Cover}, the
problem of finding at most $k$~vertices such that
their deletion removes all edges (that is, all~$K_2$s) from~$G$. %
Since the size of a smallest vertex
cover is large in many input graphs, 
parameterizations above
the lower bounds ``size of a maximum
matching~$M$ in the input graph'' and ``optimum value~$L$ of the LP~relaxation of the standard ILP-formulation of \textsc{Vertex Cover}''
have been considered.
After a series of
improvements~\cite{RO09,CPPW13,LNR+14,GP15}, the current best running time
is~$3^{\ell}\cdot n^{O(1)}$, where~$\ell:=k-(2\cdot L-|M|)$~\cite{GP15}.

We extend this approach to edge modification problems, where the number~$k$ of modifications tends to be even larger than for vertex deletion problems. For example, in the case of \textsc{Cluster Editing}, which asks to destroy induced paths on three vertices by edge modifications, the number of modifications is often larger than the number of vertices in the input graph~\cite{BBBT09}. Hence, parameterization above lower bounds
seems natural and even more relevant for edge modification problems. Somewhat surprisingly, this approach has not been considered so far.
We thus initiate research on parameterization above lower bounds in this context.
As a starting point, we focus on edge modification problems for graph properties that are characterized by one small forbidden induced subgraph~$F$:

\decprob{\textsc{$F$-free Editing}}%
{A graph~$G=(V,E)$ and a natural number~$k$.}%
{Is there an \emph{\gfedset{F}}~$S\subseteq\binom{V}{2}$ of size at most~$k$ such that~$G\symdiff{} S:=(V,(E\setminus S)\cup(S\setminus E))$ does not contain~\(F\) as induced subgraph?}
In the context of a concrete variant of \ffed{}, we refer to an~\gfedset{F} as \emph{solution} and call a solution \emph{optimal} if it has minimum size.  

\paragraph{Lower bounds from packings of bounded-cost induced subgraphs.}
Following the approach of parameterizing \textsc{Vertex Cover} above the size of a maximum matching, we can parameterize \ffed{} above a lower bound obtained from packings of induced subgraphs containing~\(F\).

\begin{definition}
  A \emph{vertex-disjoint (or edge-disjoint) packing} of induced subgraphs of a graph~$G$ is a set~$\packing=\{H_1,\ldots,H_z\}$ such that each~$H_i$ is an induced subgraph of~$G$ and such that the vertex sets (or edge sets) of the~$H_i$ are mutually disjoint.
\end{definition}
While it is natural to consider packings of $F$-graphs to obtain a lower bound on the solution size, a packing of other graphs that contain~$F$ as induced subgraph might yield better lower bounds and thus a smaller %
parameter above this lower bound. For example, a~$K_4$~contains several triangles and two edge deletions are necessary to make it triangle-free. Thus, if a graph~$G$ has a vertex-disjoint packing of~$h_3$~triangles and~$h_4$~$K_4$s, then at least~$h_3 + 2\cdot h_4$ edge deletions are necessary to make it triangle-free.\footnote{Bounds of this type are exploited, for example, in so-called cutting planes, which are used in speeding up the running time of ILP solvers. } %
Moreover, when allowing arbitrary graphs for the packing, the lower bounds provided by vertex-disjoint packings can be better than the lower bounds provided by edge-disjoint packings of~$F$. A disjoint union of~$h$~$K_4$s, for example, has~$h$ edge-disjoint triangles but also~$h$ vertex-disjoint~$K_4$s. Hence, the lower bound provided by packing vertex-disjoint~$K_4$s is twice as large as the one provided by packing edge-disjoint triangles in this graph.

Motivated by this benefit of vertex-disjoint
packings of arbitrary graphs, we mainly consider lower bounds obtained
from vertex-disjoint packings, which we assume to receive as
input. %
Thus, we arrive at the following problem, where
$\tau(G)$~denotes the minimum size of an~\gfedset{F} for a graph~$G$:
\decprob{\tffedv}%
{A graph~$G=(V,E)$, a vertex-disjoint packing~$\packing$ of induced subgraphs of~$G$ such that~$1\le \tau(H) \le t$ for each $H\in \packing$, and a natural number~$k$.}%
{Is there an \emph{\gfedset{F}}~$S\subseteq\binom{V}{2}$ of size at most~$k$ such that~$G\symdiff{} S:=(V,(E\setminus S)\cup(S\setminus E))$ does not contain~\(F\) as induced subgraph?}
The special case of \tffedv{} where only \(F\)-graphs are allowed in the packing is called \ffedv{}.

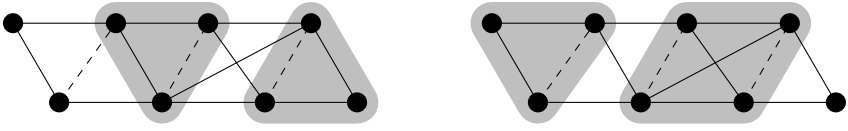
\begin{figure}[t]
  \centering
  \begin{tikzpicture}[x=0.7cm,y=0.7cm]
    \tikzstyle{packing} = [fill,color=lightgray,line cap=round, line
    join=round, line width=16pt]
    \tikzstyle{deleted} = [dashed]

    \tikzstyle{vertex} = [color=black,fill=black,circle]

    \begin{scope}[rotate=30]
      \node[vertex] (x1) at (0:1) {};
      \node[vertex] (x2) at (120:1) {};
      \node[vertex] (x3) at (-120:1) {};
    \end{scope}

    \begin{scope}[shift={(-2.8,-0.5)}]
      \begin{scope}[rotate=-30]
        \node[vertex] (z1) at (0:1) {};
        \node[vertex] (z2) at (120:1) {};
    \end{scope}
  \end{scope}
  
    \begin{scope}[shift={(2.8,-0.5)}]
      \begin{scope}[rotate=-30]
        \node[vertex] (y1) at (0:1) {};
        \node[vertex] (y2) at (120:1) {};
        \node[vertex] (y3) at (-120:1) {};
      \end{scope}
    \end{scope}
    \draw (x1.center)--(x2.center);
    \draw [deleted] (z1.center)--(x2.center);
    \draw (z1.center)--(z2.center);
    \draw (z1.center)--(x3.center);
    \draw (z2.center)--(x2.center);
    \draw (x2.center)--(x3.center) {};
    \draw [deleted] (x3.center)--(x1.center) ;
    \draw (y1.center)--(y2.center);
    \draw (y1.center)--(y3.center);
    \draw [deleted](y3.center)--(y2.center);
    \draw (x3.center)--(y3.center);
    \draw (x3.center)--(y2.center);
    \draw (x1.center)--(y3.center);
    \draw (x1.center)--(y2.center);
    \begin{pgfonlayer}{background}
      \draw[packing] (x1.center)--(x2.center)--(x3.center)--cycle; \draw[packing] (y2.center)--(y3.center)--(y1.center)--cycle;
    \end{pgfonlayer}

    \begin{scope}[shift={(9,0)}]
    \begin{scope}[rotate=30]
      \node[vertex] (x1) at (0:1) {};
      \node[vertex] (x2) at (120:1) {};
      \node[vertex] (x3) at (-120:1) {};
    \end{scope}

    \begin{scope}[shift={(-2.8,-0.5)}]
      \begin{scope}[rotate=-30]
        \node[vertex] (z1) at (0:1) {};
        \node[vertex] (z2) at (120:1) {};
    \end{scope}
  \end{scope}
  
    \begin{scope}[shift={(2.8,-0.5)}]
      \begin{scope}[rotate=-30]
        \node[vertex] (y1) at (0:1) {};
        \node[vertex] (y2) at (120:1) {};
        \node[vertex] (y3) at (-120:1) {};
      \end{scope}
    \end{scope}
    \draw (x1.center)--(x2.center);
    \draw [deleted] (z1.center)--(x2.center);
    \draw (z1.center)--(z2.center);
    \draw (z1.center)--(x3.center);
    \draw (z2.center)--(x2.center);
    \draw (x2.center)--(x3.center) {};
    \draw [deleted] (x3.center)--(x1.center) ;
    \draw (y1.center)--(y2.center);
    \draw (y1.center)--(y3.center);
    \draw [deleted](y3.center)--(y2.center);
    \draw (x3.center)--(y3.center);
    \draw (x3.center)--(y2.center);
    \draw (x1.center)--(y3.center);
    \draw (x1.center)--(y2.center);
    \begin{pgfonlayer}{background}
      \draw[packing] (x2.center)--(z1.center)--(z2.center)--cycle; \draw[packing] (y2.center)--(x1.center)--(x3.center)--(y3.center)--cycle;
    \end{pgfonlayer}
    \end{scope}

  \end{tikzpicture}
  \caption{ An instance of \textsc{Triangle Deletion}. The packing
    graphs have gray background. Left: A vertex-disjoint packing of two triangles
    giving~$\ell=1$. Right: A vertex-disjoint packing of a triangle and a~$K_4$ giving~$\ell=0$.
    The solution consists of the three dashed edges.}
\label{fig:packing-problem}
\end{figure}

\looseness=-1 From the packing~$\packing$, we obtain the lower bound $h(\packing):=\sum_{H\in\packing}\tau(H)$ on the size of an \gfedset{F}, which allows us to use the excess $\ell:=k-h(\packing)$ over this lower bound as parameter, as illustrated in \cref{fig:packing-problem}. Since~$F$ is a fixed graph, we can
compute the bound~$h(\packing)$ in~$f(t)\cdot |G|^{O(1)}$ time using the
generic algorithm~\cite{Cai96} mentioned in the introduction
for each~$H\in \packing$.
In the same time we can also verify whether the
cost-$t$ property is fulfilled.%

Packings of forbidden induced subgraphs have been used in implementations of fixed-parameter algorithms to prune the corresponding search trees tremendously~\cite{HH15}. By showing fixed-parameter algorithms for parameters above these lower bounds, we hope to %
explain the fact that these packings help in obtaining fast algorithms.

\paragraph{Our Results.}We first state the negative results since they
justify the focus on concrete problems and, to a certain extent, also
the focus on parameterizing \emph{edge} modification problems above
lower bounds obtained from
\emph{vertex}-disjoint packings. We show that \gfedv{K_6} is NP-hard
for~$\ell=0$. This proves, in particular, that a general
fixed-parameter tractability result as it is known for the
parameter~$k$~\cite{Cai96} cannot be expected. Moreover, we show that,
if~$F$ is a triangle and~$\packing$ is an \emph{edge-disjoint} packing
of~$h$ triangles in a graph~$G$, then it is NP-hard to decide
whether~$G$ has a triangle deletion set of size~$h$ (that is,
$\ell=0$). Thus, parameterization by~$\ell$ is hopeless for this
packing lower bound. We also consider vertex deletion problems. For
these we show that extending the parameterization ``above maximum
matching'' for \textsc{Vertex Cover} to \textsc{$d$-Hitting Set} in a
natural way leads to intractable problems. This is achieved by showing
that, for all~$q\ge 3$, \PLVD{} is NP-hard even if~$\ell=0$.

  Our positive results are fixed-parameter algorithms and problem kernels (a notion for provably effective polynomial-time data reduction, see \cref{sec:prelim} for a formal definition) for three variants of~\tffedv{}.  Namely, these are the variants in which~$F$~is a triangle (that is, a~$K_3$) or a
path on three vertices (that is, a~$P_3$). The first case is known as
\textsc{Triangle Deletion}, the second one as \textsc{Cluster
  Editing}.  We also consider the case in which the input is a
tournament graph and~$F$ is a directed cycle on three vertices. This
is known as \textsc{Feedback Arc Set in Tournaments}.
Using a general approach described in \cref{sec:approach}, we obtain fixed-parameter algorithms for these variants of \tffedv{} parameterized by~$t$ and~$\ell$.  This implies fixed-parameter tractability for \ffedv{} parameterized by~$\ell$.  Specifically, we obtain the following positive results:
\shorten{Applying this framework plus some problem-specific data reduction and
branching rules yields the following particular results:}
\begin{enumerate}[(i)]
\item For \textsc{Triangle Deletion}, we show an~$O((2t+3)^\ell \cdot (nm + n\cdot 2.076^t))$-time algorithm and an~$O(t\cdot
  \ell)$-vertex problem kernel for cost-$t$ packings.
\item For \textsc{Feedback Arc Set in Tournaments}, we show a
  $2^{\smash{O(\sqrt{(2t+1)\ell})}} \cdot n^{O(1)}$-time algorithm and
  an~$O(t\cdot\ell)$-vertex problem kernel for cost-$t$ packings.
\item For \textsc{Cluster Editing}, we show an~$O(1.62^{(2t+1)\cdot
    \ell} + nm + n\cdot 1.62^t)$-time algorithm and an~$O(t\cdot
  \ell)$-vertex kernel for cost-$t$ packings, and a~$4^\ell\cdot
  n^{O(1)}$-time algorithm for $P_3$-packings.
\end{enumerate}
For the kernelization results, we need to assume that $t\in O(\log n)$
to guarantee polynomial running time of the data reduction. %

\paragraph{Organization of this work.}
In \cref{sec:prelim}, we introduce basic graph-theoretic notation and formally define fixed-parameter algorithms and problem kernelization.  In \cref{sec:approach}, we present the general approach used in our algorithmic and data reduction results.  In \cref{sec:triangle-free}, we present our results regarding \textsc{Triangle Deletion}, in \cref{sec:fast} regarding \textsc{Feedback Arc Set in Tournaments}, and in \cref{sec:cluster-edit} regarding \textsc{Cluster Editing}.  \cref{sec:hardness} shows vertex and edge deletion problems that remain NP-hard for~\(\ell=0\), where \(\ell\)~is the number of modifications that are allowed in addition to a lower bound based on vertex-disjoint packings.  %
We conclude with some open questions in \cref{sec:conclusion}.

\section{Preliminaries}\label{sec:prelim}
In this section, we introduce basic graph-theoretic notation and formally define fixed-parameter algorithms and problem kernelization.
\paragraph{Notation.}
\label{sec:notation}
Unless stated otherwise, we consider undirected, simple, finite graphs~$G=(V,E)$, with a \emph{vertex set}~$V(G):=V$ and an \emph{edge set}~$E(G):=E\subseteq \binom{V}{2}:=\{\{u,v\}\mid u,v\in V \wedge u\neq v\}$.  Let~$n:=|V(G)|$ denote the \emph{order} of the graph and~$m:=|E(G)|$ its number of edges.  %
A~set~$S\subseteq \binom{V}{2}$ is an \emph{edge modification set} for~$G$. For an edge modification set~$S$ for~$G$, let~$G\symdiff{} S:=(V,(E\setminus S)\cup (S\setminus E))$ denote the \emph{graph obtained by applying}~$S$ to~$G$.  If~$S\subseteq E$, then~$S$ is called an \emph{edge deletion set} and we write~$G\setminus S$ instead of~$G\symdiff{} S$.  
The \emph{open neighborhood} of a vertex~$v\in V$ is defined as~$N_G(v):=\{u\in V\mid \{u,v\}\in E\}$. %
Also, for~$V'\subseteq V$, let~$G[V']:=(V',E\cap\binom{V'}{2})$ denote the subgraph of~$G$ \emph{induced by~$V'$}.
A \emph{directed graph (or digraph)}~$G=(V,A)$ consists of a \emph{vertex set}~$V(G)$ and an \emph{arc set}~$A(G):=A\subseteq\{(u,v)\in V^2\mid u\neq v\}$.  A \emph{tournament} on~$n$ vertices is a directed graph~$(V,A)$ with~$|V|=n$ such that, for each pair of distinct vertices~$u$ and~$v$, either~$(u,v)\in A$ or~$(v,u)\in A$.

\paragraph{Fixed-parameter algorithms.} The idea in fixed-parameter algorithms is to accept the exponential running time that seems to be inevitable when exactly solving NP-hard problems, yet to confine it to some small problem-specific parameter.  A problem is \emph{fixed-parameter tractable} with respect to some parameter~\(k\) if there is a \emph{fixed-parameter algorithm} solving any instance of size~\(n\) in $f(k)\cdot n^{O(1)}$~time.  We will also say that a problem is fixed-parameter tractable with respect to some combined parameter~``\(k\)~and~\(\ell\)'' or ``\((k,\ell)\)'' if it is fixed-parameter tractable parameterized by \(k+\ell\).

Fixed-parameter algorithms can efficiently solve instances in which the parameter~$k$ is small, even if the input size~$n$ is large.  All  vertex deletion, edge deletion, and edge modification problems for graph properties characterized by finite forbidden induced subgraphs are fixed-parameter tractable parameterized by the number of modifications~$k$~\cite{Cai96}.

\paragraph{Problem kernelization.} An important technique in
fixed-parameter algorithmics is \emph{(problem) kernelization}
\citep{Kra14}---a formal approach of describing efficient and
correct data reduction.  A \emph{kernelization} is an algorithm
that given an instance $x$ with parameter~$k$, yields an instance $x'$
with parameter~$k'$ in time polynomial in $|x|+k$ such that $(x,k)$~is
a yes-instance if and only if $(x',k')$~is a yes-instance, and if both
$|x'|$ and $k'$ are bounded by some functions $g$ and $g'$ in~$k$,
respectively. The function~$g$ is referred to as the \emph{size} of
the \emph{problem kernel}~$(x',k')$. Kernelizations are commonly
described by giving a set of \emph{data reduction rules} which when
applied to an instance~$x$ of a problem yield an instance~$x'$. We say
that a data reduction rule is \emph{correct} if~$x$ and~$x'$ are equivalent.

All vertex deletion problems for establishing graph properties
characterized by a finite number of forbidden induced subgraphs have a
problem kernel of size polynomial in the parameter~\(k\) of allowed
modifications~\citep{Kra12}.  In contrast, many variants of~\ffed{} do
not admit a problem kernel whose size is polynomial
in~$k$~\cite{KW13,GHPP13,CC15}. 

\section{General Approach}\label{sec:approach}

In this section, we describe the general approach of our
fixed-parameter algorithms. Recall that $\tau(H)$~is the minimum
number of edge modifications required to transform a graph~$H$ into
an~\(F\)-free graph.  We present fixed-parameter algorithms for three
variants of \tffedv{} parameterized by the combination of~$t$
and~$\ell:=k-h(\packing)$,
where~$h(\packing):=\sum_{H\in\packing}\tau(H)$. The idea behind the
algorithms is to arrive at a classic
win-win scenario~\cite{Fell03} where we can either apply data
reduction or show that the packing size~$|\packing|$ is bounded.  This
will allow us to bound~\(k\) in~\(t\cdot \ell\) for yes-instances and,
thus, to apply known fixed-parameter algorithms for the
parameter~\(k\) to obtain fixed-parameter tractability results
for~\((t,\ell)\).

More precisely, we show that, for each induced
subgraph~\(H\) of~\(G\) in a given packing~\(\packing\), we face
essentially two situations. If there is an optimal solution for~$H$
that is a subset of an optimal solution for~\(G\), then we can apply a data reduction rule. Otherwise, we find a
certificate witnessing that $H$ itself needs to be solved
suboptimally or that a vertex pair containing exactly
one vertex from~$H$ needs to be modified. We use the following terminology for these pairs. 
\begin{definition}[External vertex pairs and edges]A vertex pair~$\{u,v\}$ is an \emph{external pair for a packing graph~$H\in \packing$} if exactly one of~$u$ or~$v$
  is in~$V(H)$, an edge is an \emph{external edge for~$H\in \packing$} if exactly one of its endpoints
  is in~$H$. 
\end{definition}
Observe that every pair or edge is an external pair or edge for at most two packing graphs
since the packing graphs are vertex-disjoint. Therefore, the modification of an external vertex
pair can destroy at most two certificates.
This is the main fact used in the proof of the
following bound on~$\ell$.

\begin{lemma}\label[lemma]{lem:gen-bound}
  Let~$(G,\packing,k)$ be an instance of \tffedv{} and let $S$~be a
  size-$k$ solution that contains, for each~$H=(W,F)\in \packing$,
  \begin{enumerate}[(a)]
  \item\label{inside} at least $\tau(H)+1$ vertex pairs from~$\binom{W}{2}$, or
  \item\label{outside} at least one external vertex pair~$\{v,w\}$ for~$H$.
  \end{enumerate}
  Then, $|\packing|\le 2\ell$ and thus, $k\le (2t+1)\ell$.
\end{lemma}
\begin{proof}
    Denote by~$\packing_a\subseteq\packing$ the set of all graphs in~$\packing$ that fulfill
  property~(\ref{inside}) and let $p_a:=|\packing_a|$.
  Let $\packing_b:=\packing\setminus \packing_a$ denote the
  set containing the remaining packing graphs (fulfilling property~(\ref{outside})) and let $p_b:=|\packing_b|$.
  Thus, $|\packing|=p_a + p_b$.
  Furthermore, let $h_a:=\sum_{H\in\packing_a}\tau(H)$~denote the lower bound obtained from the graphs in~$\packing_a$ and let $h_b:=h(\packing)-h_a$~denote the part of the lower bound obtained by the remaining graphs.

  The packing graphs in~$\packing_a$ cause~$h_a+p_a$ edge modifications inside of them. Similarly, the packing graphs in~$\packing_b$ cause at least~$h_b$ edge modifications inside of them, and each packing graph~$H\in \packing_b$ additionally causes modification of at least one external vertex pair for~$H$. Since every vertex pair is an external pair for at most two different packing graphs, at least~$h_b+p_b/2$ edge modifications are caused by the graphs in~$\packing_b$.  This implies that
  \begin{align*}\allowdisplaybreaks
    && k&\ge h_a+h_b+p_a+p_b/2& \\
    \Leftrightarrow && k-h(\packing)&\ge p_a+p_b/2& \\
    \Leftrightarrow && 2\ell&\ge 2p_a+p_b\ge |\packing|.
  \end{align*}
Consequently,  $k = \ell+h(\packing)\le \ell + t\cdot
  |\packing|\le \ell + t\cdot 2\ell=(2t+1)\ell$ .\qed
\end{proof}

\section{Triangle Deletion}
\label{sec:triangle-free}

In this section, we study \textsc{Triangle Deletion}, the problem of destroying all \emph{triangles}~($K_3$s) in a graph by at most $k$~edge deletions.  In \cref{sec:trdpos}, we apply our framework from \cref{sec:approach} to show that \textsc{Triangle Deletion} is fixed-parameter tractable parameterized above the lower bound given by a cost-\(t\) packing.  In \cref{sec:trdneg}, we then show that parameterization above a lower bound given by edge-disjoint packings of triangles does not lead to fixed-parameter algorithms unless P\({}={}\)NP.

\subsection{A fixed-parameter algorithm  for vertex-disjoint cost-\(t\) packings}\label{sec:trdpos}
Before presenting our new fixed-parameter tractability results for \textsc{Triangle Deletion}, let us first summarize the known results concerning the (parameterized)
complexity of \textsc{Triangle Deletion}.
\textsc{Triangle Deletion} is NP-complete~\cite{Yan81}. %
It allows for a trivial reduction to \textsc{3-Hitting Set} since edge deletions do not create new triangles~\cite{GGHN04}. Combining this approach with the currently fastest known algorithms for \textsc{3-Hitting Set}~\cite{wahl07,Bev14} gives an algorithm for \textsc{Triangle Deletion} with running time~$O(2.076^k+nm)$. Finally, \textsc{Triangle Deletion} admits a problem kernel with at most $6k$~vertices~\cite{BKM09}.
  We show that \TAGP{} is fixed-parameter tractable with respect to
the combination of~$t$ and~$\ell:=k-h(\packing)$.
More precisely, we obtain a kernelization and a search tree algorithm.
Both make crucial use of the following generic reduction rule for \TAGP.

\begin{rrule}\label[rrule]{rule:opt-graph}
  If there is an induced subgraph~\(H\in\mathcal H\) and a set~$T\subseteq E(H)$ of~$\tau(H)$ edges such that deleting~$T$ destroys all triangles of~$G$ that contain edges of~$H$, then delete~$T$ from~$G$, \(H\)~from~\(\mathcal H\) and decrease~$k$ by~$\tau(H)$.
\end{rrule}
\begin{lemma}\label{lem:opt-graph-correct}
  \cref{rule:opt-graph} is correct. 
\end{lemma}
\begin{proof}
  Let~$(G,\packing,k)$ be the instance
  to which \cref{rule:opt-graph} is applied and let $(G',\packing\setminus\{H\},k-\tau(H))$ with \(G':=G\setminus T\) be the result. We show that~$(G,\packing,k)$ is a yes-instance if and only if $(G',\packing\setminus\{H\},k-\tau(H))$~is.

  First, let~$S$ be a solution of size at most~$k$ for~$(G,\packing,k)$. Let~$S_H:=S\cap E(H)$ denote the set of edges of~$S$ that destroy all triangles in~$H$. By definition, $|S_H|\ge \tau(H)$.  Since~$S_H\subseteq E(H)$, only triangles containing at least one edge of~$H$ are destroyed by deleting~$S_H$. It follows that the set of triangles destroyed by~$S_H$ is a subset of the triangles destroyed by~$T$.  Hence, $(S\setminus S_H)\cup T$ has size at most~$k$ and clearly is a solution for~$(G,\packing,k)$ that contains all edges of~$T$. Thus, \((S\setminus S_H)\) is a solution of size~\(k-\tau(H)\) for \(G\setminus T=G'\) and $(G',\packing\setminus\{H\},k-\tau(H))$ is a yes-instance.

  For the converse direction, let~$S'$ be a solution of size at most~$k-\tau(H)$ for~$(G',\packing\setminus\{H\},k-\tau(H))$. Since~$T\subseteq E(H)$, it holds that every triangle contained in~$G$ that does not contain any edge of~$H$ is also a triangle in~$G'$.  Thus, $S'$~is a set of edges whose deletion in~$G$ destroys all triangles that do not contain any edge of~$H$. Since~$T$ destroys all triangles containing an edge of~$H$, we have that~$T\cup S'$ is a solution for~$G$.  Its size is~\(k\).\qed
\end{proof}

\noindent We now show that, if \cref{rule:opt-graph} is not applicable to~$H$, then we can find a \emph{certificate} for this, which will allow us later to branch efficiently on the destruction of triangles:
\begin{definition}[Certificate]
  A \emph{certificate} for inapplicability of \cref{rule:opt-graph} to an induced subgraph~\(H\in\packing\) is a set~$\mathcal T$ of triangles in~$G$, each containing exactly one distinct edge of~$H$, such that $|\mathcal T|=\tau(H)+1$ or~$|\mathcal T|\le\tau(H)$ and~$\tau(H')>\tau(H)-|\mathcal T|$, where~$H'$ is the subgraph obtained from~$H$ by deleting, for each triangle in~$\mathcal{T}$, its edge shared with~$H$.
\end{definition}
\begin{lemma}\label[lemma]{lem:opt-graph-time}
  Let $\Time(G,k)$~be the time needed to compute a triangle-free
  deletion set of size at most~\(k\) in a graph~\(G\) if it exists.

  In $O(nm+\sum_{H\in \packing} \Time(H,t))$~time, we can apply
  \cref{rule:opt-graph} to all~$H\in \packing$ and output a
  certificate~$\mathcal T$ if \cref{rule:opt-graph} is inapplicable to some~$H\in
  \packing$.
\end{lemma} 
\noindent
In the statement of the lemma,
we assume that
$\Time$~is monotonically nondecreasing
in the size of~$G$ and in~$k$.
As described above,
currently $O(2.076^k+|V(G)|\cdot |E(G)|)$
is the best known bound for~$\Time(G,k)$.
\begin{proof}[of \cref{lem:opt-graph-time}]
  First, in~$O(nm)$ time, we compute for all~$H\in\packing$ all triangles~$\mathcal T$ that contain exactly one edge~$e\in E(H)$. These edges are labeled in each~$H\in \packing$.  Then, for each $H\in \packing$, in $\Time(H,t)$~time we determine the size~$\tau(H)$ of an optimal triangle-free deletion set for~$H$. Let $t'$~denote the number of labeled edges of~$H$.  

\emph{Case~1:~$t'>\tau(H)$.} In this case, we return as certificate $\tau(H)+1$~triangles of~$\mathcal{T}$, each containing a distinct of $\tau(H)+1$~arbitrary labeled edges.  

\emph{Case~2:~$t'\le \tau(H)$.} Let~$H'$~denote the graph obtained from~$H$ by deleting the labeled edges. All triangles of~$G$ that contain at least one edge of~$H$ either contain a labeled edge or they are contained in~$H'$. Thus, we now determine in~$\Time(H',\tau(H)-t')$ time whether~$H'$ can be made triangle-free by $\tau(H)-t'$~edge deletions. If this is the case, then the rule applies and the set~$T$ consists of the solution for~$H'$ plus the deleted labeled edges.  Otherwise, destroying all triangles that contain exactly one edge from~$H$ leads to a solution which needs more than~$\tau(H)$~edge deletions and thus the rule does not apply. In this case, we return the certificate~$\mathcal T$ for this~$H\in \packing$. 

The overall running time now follows from the monotonicity of~$f$,
from the fact that~$|\packing|\le n$, and from the fact that one pass
over~$\packing$ is sufficient since deleting edges in each~$H$ does
not produce new triangles and does not destroy triangles in
any~$H'\neq H$.  \qed\end{proof}

\noindent %
Observe that \cref{rule:opt-graph} never increases the parameter~$\ell$ since we decrease both~$k$ as well as the lower bound~$h(\packing)$ by~$\tau(H)$.
After application of \cref{rule:opt-graph}, we can upper-bound
the solution size~$k$ in terms of~$t$ and~$\ell$, which allows us to transfer parameterized complexity results for the parameter~\(k\) to the combined parameter~\((t,\ell)\).
\begin{lemma}\label[lemma]{lem:k-small-gen}
  Let~$(G,\packing,k)$ be a yes-instance of~\TAGP{} such that \cref{rule:opt-graph} is inapplicable. Then, $k\le (2t+1)\ell$.
\end{lemma}
\begin{proof}
  Since $(G,\packing,k)$~is reduced with respect to \cref{rule:opt-graph}, for each graph~$H=(W,F)$ in~$\packing$, there is a set of edges between~$W$ and $V\setminus W$ witnessing that every optimal solution for~$H$ does not destroy all triangles containing at least one edge from~$H$.  Consider any optimal solution~$S$. For each graph~$H\in\mathcal H$, there are two possibilities: Either at least~$\tau(H)+1$ edges inside~$H$ are deleted by~$S$, or at least one external edge of~$H$ is deleted by~$S$.  Therefore, $S$~fulfills the condition of \cref{lem:gen-bound} and thus~$k\le (2t+1)\ell$.
\qed\end{proof}

\begin{theorem}\label{thm:td-kern-general}
Let $\Time(G,k)$~be the time used for computing a triangle-free deletion set of size at most~\(k\) in a graph~\(G\) if it exists. Then, \TAGP{}
  \begin{enumerate}[(i)]
  \item\label{td-searchtree} can be solved in~$O((2t+3)^\ell \cdot (nm + \sum_{H\in \packing} \Time(H,t)))$~time, and
  \item\label{td-kern} admits a problem kernel with at most~$(12t+6) \ell$ vertices that can be computed in $O(nm+\sum_{H\in \packing} \Time(H,t))$~time.
  \end{enumerate}
\end{theorem}
\begin{proof}
  We first prove \eqref{td-kern}.  To this end, let $(G=(V,E),\packing,k)$ be the input instance.  First, compute in $O(nm+\sum_{H\in \packing} \Time(H,t))$~time an instance that is reduced with respect to \cref{rule:opt-graph}. Afterwards, by \cref{lem:k-small-gen}, we can reject if~$k > (2t+1)\ell$.  Otherwise, we apply the known kernelization algorithm for \textsc{Triangle Deletion} to the instance~$(G,k)$ (that is, without~$\packing$). This kernelization produces in $O(m\sqrt{m})=O(nm)$~time a problem kernel~$(G',k')$ with at most~$6k\le (12t+6)\ell$ vertices and with~$k'\le k$~\cite{BKM09}. Adding an empty packing gives an equivalent instance~$(G',\emptyset,k')$ with parameter~$\ell'=k'\le (2t+1)\ell$ of \TAGP.

  It remains to prove \eqref{td-searchtree}.  To this end, first apply \cref{rule:opt-graph} exhaustively in $O(nm + n\cdot \Time(H,t))$~time.  Now, consider a reduced instance. If~$\ell<0$, then we can reject the instance. Otherwise, consider the following two cases.
  
  \emph{Case 1: $\packing=\emptyset$.} If~$G$ is triangle-free, then we are done. Otherwise, pick an arbitrary triangle in~$G$ and add it to~$\packing$.

  \emph{Case 2: $\packing$ contains a graph~$H$.} Since \cref{rule:opt-graph} does not apply to~$H$, there is a certificate~$\mathcal{T}$ of~$t'\le \tau(H)+1$ triangles, each containing exactly one distinct edge of~$H$ such that deleting the edges of these triangles contained in $H$ produces a subgraph~$H'$ of~$H$ that cannot be made triangle-free by~$\tau(H)-t'$ edge deletions.  Thus, branch into the following~$(2t'+1)$ cases: First, for each triangle~$T\in \mathcal{T}$, create two cases, in each deleting a different one of the two edges of~$T$ that are not in~$H$. In the remaining case, delete the $t'$~edges of~$H$ and replace~$H$ by~$H'$ in~$\packing$.

  It remains to show the running time by bounding the search tree size. In Case~1, no branching is performed and the parameter is decreased by at least one. In Case~2, the parameter value is decreased by one in each branch: in the first~$2t'$ cases, an edge that is not contained in any packing graph is deleted. Thus, $k$~decreases by one while $h(\packing)$~remains unchanged. In the final case, the value of~$k$ decreases by~$t'$ since this many edge deletions are performed. However, $\tau(H')\ge\tau(H)-t'+1$.  Hence, the lower bound~$h(\packing)$ decreases by at most~$t'-1$ and thus the parameter~$\ell$ decreases by at least one.  Note that applying \cref{rule:opt-graph} never increases the parameter. Hence, the depth of the search tree is at most~$\ell$.  \qed\end{proof} 

\begin{corollary}\label[corollary]{cor:triangle-abv-cost-t}
\TAGP{}
  \begin{enumerate}[(i)]
  \item\label{cor-td-searchtree} can be solved in~$O((2t+3)^\ell \cdot (nm + n\cdot 2.076^t)$~time, and
  \item\label{cor-td-kern} admits a problem kernel with at most~$(12t+6) \ell$ vertices that can be computed in~$O(nm+n\cdot 2.076^t)$ time.
  \end{enumerate}
\end{corollary}
For the natural special case~$t=1$, that is, for triangle packings, \cref{thm:td-kern-general}\eqref{cor-td-searchtree} immediately yields the following running time.
\begin{corollary}\label[corollary]{cor:triangle-abv}
  \TATP{} is solvable in $O(5^{\ell}nm)$~time.
\end{corollary}

\subsection{Hardness for edge-disjoint packing}\label{sec:trdneg}
\noindent We complement the positive results of \cref{thm:td-kern-general} by the following hardness result for the case of edge-disjoint triangle packings:
\begin{theorem}\label[theorem]{thm:girlanden}
  \textsc{Triangle Deletion} is NP-hard even for \(\ell:=k-|\packing|=0\) if \(\packing\)~is an \emph{edge-disjoint} packing of triangles.
\end{theorem}

\noindent \cref{thm:girlanden} shows that parameterizing \textsc{Triangle Deletion} over a lower bound given by edge-disjoint packings cannot lead to fixed-parameter algorithms unless P\({}={}\)NP.  We prove \cref{thm:girlanden} using a reduction from~\POITS{}.

\decprob{\POITS}
{\label[problem]{prob:poits}A Boolean formula $\phi=C_1\wedge\ldots\wedge C_m$ in conjunctive normal form over variables~$x_1,\ldots,x_n$ with at most three variables per clause.}
{Does $\phi$ have a satisfying assignment?}
\begin{construction}\label[construction]{cons:girlanden}
  Given a Boolean formula~$\phi$, we create a graph~$G$ and an %
  edge-disjoint packing~$\packing$ of triangles such that $G$~can be made triangle-free by exactly $|\packing|$~edge deletions if and only if there is a satisfying assignment for~$\phi$.
  We assume that each clause of~$\phi$ contains exactly three pairwise distinct variables.  The construction is illustrated in \cref{fig:k3del-hard}.
\begin{figure}[t]
  \centering
  \includegraphics{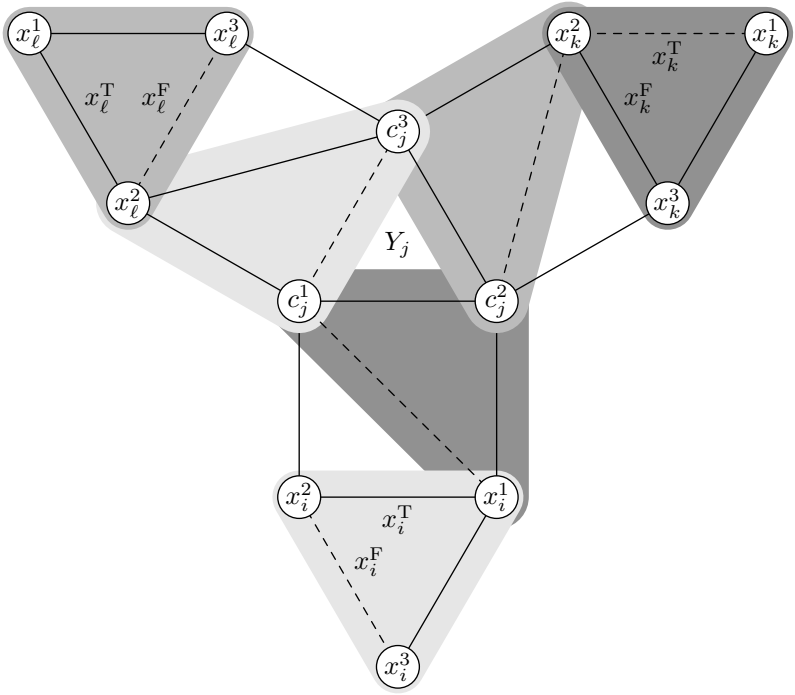}
  \caption{Construction for a clause \(C_j=(x_i\vee \neg x_k\vee \neg x_\ell)\).  The triangles on a gray background are contained in the mutually edge-disjoint triangle packing~$\packing$.  Deleting the dashed edges corresponds to setting~$x_i$ and~\(x_\ell\) to false and \(x_k\)~to true, thus satisfying~$C_j$.  Note that it is impossible to destroy triangle~$Y_j$ by $|\packing|$~edge deletions if we delete~$x_i^\text{F}$, $x_k^\text{T}$, and $x_\ell^\text{T}$, which corresponds to the fact that clause~\(C_j\) cannot be satisfied by this variable assignment.}
\label{fig:k3del-hard}
\end{figure}
  
For each variable~$x_i$ of~$\phi$, create a triangle~$X_i$ on the vertex set~$\{x_i^1,x_i^2,x_i^3\}$ with two distinguished edges~$x_i^\text{T}:=\{x_i^1,x_i^2\}$ and~$x_i^\text{F}:=\{x_i^2,x_i^3\}$ and add~$X_i$ to~$\packing$.  For each clause~$C_j=(l_1,l_2,l_3)$ of~$\phi$, create a triangle~$Y_j$ on the vertex set~$\{c_j^1,c_j^2,c_j^3\}$ with three edges~$c_j^{l_1}$, $c_j^{l_2}$, and~$c_j^{l_3}$.  Connect the clause gadget~$Y_j$ to the variable gadgets as follows: If $l_t=x_i$, then connect the edge $c_j^{l_t}=:\{u,v\}$ to the edge~$x_i^\text{T}=\{x_i^1,x_i^2\}$ via two adjacent triangles~$A_{ij}:=\{u,v,x_i^1\}$ and $B_{ij}:=\{v,x_i^1,x_i^2\}$ sharing the edge~$\{v,x_i^1\}$.  The triangle~$A_{ij}$ is added to~$\packing$.  If $l_t=\neg x_i$, then connect the edge $c_j^{l_t}=:\{u,v\}$ to the edge~$x_i^\text{F}=\{x_i^2,x_i^3\}$ via two adjacent triangles~$A_{ij}:=\{u,v,x_i^3\}$ and $B_{ij}:=\{v,x_i^2,x_i^3\}$ sharing the edge~$\{v,x_i^3\}$.  The triangle~$A_{ij}$ is added to~$\packing$.
\end{construction}

\begin{proof}[of \cref{thm:girlanden}]
  First, observe that \cref{cons:girlanden} introduces no edges between distinct clause gadgets or distinct variable gadgets.  Thus, under the assumption that each clause contains each variable at most once, the only triangles in the constructed graph are the~$X_i$, the~$Y_j$, the~$A_{ij}$ and~$B_{ij}$ for all variables~$x_i$ and the incident clauses~$C_j$.

  Now, assume that $\phi$~allows for a satisfying assignment.  We construct a set of edges~$S$ of size~$|\packing|$ such that $G':=(V,E\setminus S)$ is triangle-free.  For each variable~$x_i$ that is true, add $x_i^\text{T}$ to~$S$.  For each variable~$x_i$ that is false, add $x_i^\text{F}$ to~$S$.  By this choice, the triangle~$X_i$ is destroyed in~$G'$ for each variable~$x_i$.  Additionally,
  for each clause~$C_j$ and its \emph{true} literals~$l\in\{x_i,\neg x_i\}$, the triangle~$B_{ij}$ is destroyed. To destroy~$A_{ij}$, we add to~$S$ the edge of~$A_{ij}$ shared with~$Y_j$, which also destroys the triangle~$Y_j$.
  For each clause~$C_j$ containing a \emph{false} literal~$l\in\{x_i,\neg x_i\}$, we destroy~$B_{ij}$ and simultaneously~$A_{ij}$ by adding to~$S$ the edge of~$A_{ij}$ shared with~$B_{ij}$.

  Conversely, assume that there is a set~$S$ of size~$|\packing|$ such that $G'=(V,E\setminus S)$ is triangle-free.  We construct a satisfying assignment for~$\phi$.
  First, observe that, since the triangles in $\packing$ are pairwise edge-disjoint, $S$~contains exactly one edge of each triangle in~$\packing$.  Thus, of each triangle~$X_i$, at most one of the two edges~$x_i^\text{F}$ and~$x_i^\text{T}$ is contained in~$S$.
  The set~$S$ contains at least one edge~$e$ of each~$Y_j$.  This edge is shared with a triangle~$A_{ij}$.  Since $A_{ij}\in\packing$ and, with~$e$, $S$~already contains one edge of~$A_{ij}$, $S$~does not contain the edge shared between~$A_{ij}$ and~$B_{ij}$.  Since~$B_{ij}\notin\packing$, $S$~has to contain an edge of~$B_{ij}$ shared with another triangle in~$\packing$.  If the clause~$C_j$ contains~$x_i$, then the only such edge is~$x_i^\text{T}$ and we set~$x_i$ to true.  If the clause~$C_j$ contains~$\neg x_i$, then the only such edge is~$x_i^\text{F}$ and we set~$x_i$ to false.  In both cases, clause~$C_j$ is satisfied.  Since at most one of~$x_i^\text{T}$ and~$x_i^\text{F}$ is in~$S$, the value of each variable~$x_i$ is well-defined.
\qed\end{proof}

\section{Feedback Arc Set in Tournaments}
\label{sec:fast}
 
In this section, we present a fixed-parameter algorithm and a problem kernel for \textsc{Feedback Arc Set in Tournaments} parameterized above lower bounds of cost-\(t\) packings. 

In \textsc{Feedback Arc Set in Tournaments}, we are given a directed tournament graph~$G$
as input and want to delete a minimum number of arcs to make the graph acyclic, that is, to destroy all directed cycles
in~$G$. Due to a well-known observation, one can also view \textsc{Feedback Arc Set in Tournaments} as an arc \emph{reversal} problem: After deleting a minimum set of arcs to make the graph acyclic, adding the
arc~$(u,v)$ for every deleted arc~$(v,u)$ does not create any cycle. Since, in
tournaments, every pair of vertices is connected by exactly one arc, it follows that that
destroying cycles by edge deletions is equivalent to destroying them by arc
reversals. Altogether, we arrive at the following problem definition. 
\decprob{\FAST{} (\textsc{FAST})}
{An $n$-vertex tournament~$G=(V,A)$ and a natural number~$k$.}
{Does~$G$ have a \emph{feedback arc set}~$S\subseteq A$, that is, a set~$S$ such that reversing all arcs in~$S$ yields an acyclic tournament, of size at most~$k$?}
\textsc{FAST} is NP-complete~\cite{Alon06} but fixed-parameter
tractable with respect to~$k$~\cite{RS06, DGHNT06, ALS09,Fei09,KS10,FP13}. The running
time of the current best fixed-parameter algorithm is~$2^{c\cdot \sqrt{k}}
+ n^{O(1)}$ where~$c\le 5.24$~\cite{FP13}.
Moreover, a problem kernel with $(2+\epsilon)k$ vertices
for each \emph{constant}~$\epsilon>0$ is known~\cite{BFGPPST11}
as well as a simpler~$4k$-vertex kernel~\cite{PPT16}.
It is well-known that a tournament is acyclic if
and only if it does not contain a \emph{directed triangle} (a cycle on 3
vertices). Hence, the problem is to find a set of arcs whose reversal
leaves no directed triangle in the tournament.

We show fixed-parameter tractability of FAST\textsc{ with Cost-$t$ Packing} parameterized by the combination of~$t$ and~$\ell:=k - h(\packing)$. %
Recall that $h(\packing):=\sum_{H\in\packing}\tau(H)\ge |\packing|$, where~$\tau(G)$ is the size of a minimum feedback arc set for a directed graph~$G$.
The approach is the same as for \textsc{Triangle Deletion} in \cref{sec:triangle-free}, that is, we upper-bound the solution size~$k$ in~$t$ and~$\ell$ and apply the fixed-parameter algorithm for~$k$ \cite{KS10}. Observe in this context that \cref{lem:gen-bound} is also correct if the input graphs are directed and if a solution contains arc reversals, since we observed arc reversals and deletions to be equivalent in the context of \textsc{FAST}.\footnote{For directed input graphs, we use the term \emph{external arc} instead of \emph{external edge}.}
We use the following reduction rule for \textsc{FAST} analogous to \cref{rule:opt-graph} for \textsc{Triangle Deletion}.

\begin{rrule}\label[rrule]{fast-general-rule}
  If there is a subtournament~$H\in\packing$ 
  and a feedback arc set~$T\subseteq A(H)$ of size~$\tau(H)$ such that
  reversing the arcs in~$T$ leaves no directed triangles in~$G$ containing arcs
  of~$H$, then reverse the arcs in~$T$, remove \(H\) from~\(\packing\), and decrease~$k$ by~$\tau(H)$.
\end{rrule}

\noindent Although \cref{fast-general-rule} is strikingly similar to \cref{rule:opt-graph}, its correctness proof is significantly more involved.

\begin{lemma}\label[lemma]{fast-general-rule-time}
  \cref{fast-general-rule} is correct and, given the tournaments~$G$ and~$H$ it
  can be applied in~$O\left(\binom{q}{2}(n-q)+\Time(H,t)\right)$ time, where $q:=|V(H)|$ and~$\Time(H,t)$ denotes the running time needed to compute a feedback arc set of size at most~$t$ in~$H$ if it exists.
\end{lemma}
\begin{proof}
  We first show correctness.
  Let~$I':=(G',\packing\setminus\{H\},k-\tau(H))$ be the instance
  created by \cref{fast-general-rule}
  from \(I:=(G,\packing,k)\)
  by reversing a subset~\(T\) of arcs
  of a subtournament~$H\in\packing$ of~\(G\).
  If \(I'\)~is a yes-instance,
  then so is~\(I\)
  since \(G'\)~is the graph~\(G\)
  with the \(\tau(H)\)~arcs in~\(T\) reversed
  and, thus,
  adding these arcs to a feedback arc set
  of size~\(k-\tau(H)\) for~\(G'\)
  gives a feedback arc set of size~\(k\) for~\(G\).
  It remains to prove that
  if $I$~is a yes-instance,
  then so is~\(I'\).
  To this end, we show that
  there is a minimum-size feedback arc set~$S$
  for~$G$ with~$T\subseteq S$.
  
  Let~$S$ be a minimum-size feedback arc set for~$G=(V,E)$.
  This implies the existence
  of a linear ordering~$\sigma_S = v_1,\ldots,v_n$
  of the vertices~$V$
  such that there are~$|S|$ \emph{backward arcs},
  that is,
  arcs~$(v_i,v_j)\in A$ such that~$i > j$.
  Now, let~$\sigma_T=w_1,\ldots,w_{|W|}$ be the ordering
  of the vertices of~\(H=(W,F)\)
  corresponding to the local solution~$T$ for~$H$
  with $\tau(H)$~backward arcs.
  Let~$N^+(w):=\{(w,v)\in A\mid v\in V\setminus W\}$ denote the \emph{out-neighbors} in~$V\setminus W$ of a vertex~$w\in W$. Analogously,
  $N^-(w):=\{(v,w)\in A\mid v\in V\setminus W\}$~denotes 
  the set of \emph{in-neighbors}.
  By the assumption of the rule, for all~$i < j$,
  \begin{equation}
    N^+(w_j) \subseteq N^+(w_i) \text{ and } N^-(w_i) \subseteq N^-(w_j)
  \end{equation}
  holds since otherwise,
  after reversing the arcs in~$T$,
  there exists a directed cycle containing an arc of~$H$
  (because the arc~$(w_i,w_j)$ is present).
  
  If the vertices of~$W$ appear in~$\sigma_S$ in the same relative order as in~$\sigma_T$, then we have~$T\subseteq S$ and we are done.
  Otherwise, we show that
  we can swap the positions of vertices of~$W$ in~\(\sigma_S\)
  so that their relative order is the same as in~$\sigma_T$
  without increasing the number of backward arcs.

  First, note that the number of backward arcs
  between vertices in~$V\setminus W$
  does not change
  when only swapping positions of vertices in~\(W\).
  Also, by assumption,
  the number of backward arcs between vertices in~$W$
  in any ordering is at least~$\tau(H)$,
  whereas it is exactly \(\tau(H)\)
  when ordering them according to~\(\sigma_T\).
  Thus,
  it remains to show that the number of backward arcs
  between vertices in~$W$ and~$V\setminus W$
  is not increased.
  To this end,
  consider a series of \emph{swaps}
  of pairs of vertices~$w_j$ and~$w_i$ such that $i<j$,
  where $w_j$ appears before~$w_i$ in~$\sigma_S$,
  reordering the vertices in~$W$ according to~$\sigma_T$.
  Let~$Y$ denote the set of all vertices that lie between~$w_j$ and~$w_i$ in~$\sigma_S$.
  Note that swapping~$w_j$ and~$w_i$ removes the backward arcs from~$w_i$ to the vertices in~$N^+(w_i)\cap Y$ and the backward arcs from vertices in~$N^-(w_j)\cap Y$ to~$w_j$, whereas it introduces new backward arcs from~$w_j$ to~$N^+(w_j)\cap Y$ and from~$N^-(w_i)\cap Y$ to~$w_i$.
  However, by the inclusions in~(1), it follows that the overall number of backward arcs does not increase in each swap.
  Hence, the overall number of backward arcs is not increased by repositioning the vertices in~$W$
  according to~$\sigma_T$. It follows that there is an optimal solution containing~$T$.
    
  It remains to show the running time.
  First,
  in $\Time(H,\tau(H))$~time,
  we compute the size~$\tau(H)$
  of an optimal feedback arc set for~$H=(W,F)$.
  Now,
  for each arc~$(u,v)\in F$,
  we check whether there is a vertex~$w\in V\setminus W$
  that forms a directed triangle
  with~$u$ and~$v$.
  If~such a vertex exists,
  then we reverse the arc~$(u,v)$.
  If this arc reversal introduces a new directed triangle
  with another vertex from~$V\setminus W$,
  then the rule does not apply.
  Overall, this procedure requires~$O\big(|F|\cdot(n-|W|)\big)$ time.
  Let~$T^*$~denote the set of arcs that are reversed in this process. Clearly, if~$|T^*| > \tau(H)$, then the rule does not apply. Otherwise, let $H'$~denote the graph obtained from~$H$ by reversing the arcs in~$T^*$ and observe that each remaining directed triangle of~$G$ that contains at least one arc of~$H'$ is contained in~$H'$.
  Thus, we now compute whether~$H'$ has a feedback arc set~$T'$ of size~$\tau(H)-|T^*|$ in~$\Time(H',\tau(H)-|T^*|)$ time. If this is the case, then the rule applies
  and we set~$T := T' \cup T^*$ (note that~$T'\cap T^*=\emptyset$, since otherwise~$|T|< \tau(H)$, which is not possible by definition of~$\tau(H)$).
  Otherwise, removing all directed triangles that contain at least one arc from~$H$ requires more than~$\tau(H)$ arc reversals and thus the rule does not apply.
\qed\end{proof}

\noindent Exhaustive application of \cref{fast-general-rule} allows us to show that~$k\le (2t+1)\ell$ holds for any yes-instance (analogous to \cref{lem:k-small-gen}).

\begin{lemma}\label[lemma]{fast-k-small-gen}
  Let~$(G,\packing,k)$ be a yes-instance of \FASTGP{}
  such that \cref{fast-general-rule} cannot be applied
  to any tournament in~$\packing$.
  Then, $k\le (2t+1)\ell$.
\end{lemma}

\begin{proof}
  Since \cref{fast-general-rule} cannot be applied
  to any tournament in~$\packing$,
  for each tournament~$H=(W,F)$ in~$\packing$,
  there is a set of arcs
  between~$W$ and~$V\setminus W$
  that witness that
  no optimal feedback arc set for~$H$
  removes all directed triangles
  containing at least one arc from~$F$.
  
  Now, for any optimal solution,
  there are two possibilities
  for each packing tournament~$H\in\packing$:
  \begin{enumerate}[(a)]
  \item\label{fast-inside} at least~$\tau(H)+1$ arcs in~$H$ are reversed, or
  \item\label{fast-outside} at least one external arc of~$H$ is reversed.
  \end{enumerate}
  Therefore, $S$~fulfills the condition
  of \cref{lem:gen-bound}
  and, thus,
  $|\packing|\le 2\ell$.
\qed\end{proof}

\noindent The bound on~$k$ yields the following two fixed-parameter tractability results.

\begin{theorem}\label{fast-general-fpt}
  Let $\Time(G,k)$~be the running time used for finding a minimum feedback arc set of size at most~\(k\) for a given tournament~$G$ if it exists.  Then, \FASTGP{}
  \begin{enumerate}[(i)]
  \item\label{fvs-fpt} is solvable in $\Time(G,(2t+1)\ell) + n^{O(1)} + \sum_{H\in\packing}\Time(H,t)$ time, and
  \item\label{fvs-kern} admits a problem kernel with at most~$(8t+4)\ell$ vertices computable in~$n^{O(1)} + \sum_{H\in\packing}\Time(H,t)$ time.
  \end{enumerate}
\end{theorem}

\begin{proof}
  \eqref{fvs-fpt} Given $(G,\packing,k)$, we first apply \cref{fast-general-rule} for each~$H\in \packing$.  This application can be performed in~$O(|\packing|n^3 + \sum_{H\in\packing}\Time(H,t))$ time by \cref{fast-general-rule-time} since~$|\packing|\le n$.
  One pass of this rule is sufficient
  to obtain an instance that is reduced:
  reversing arcs in some~$H\in \packing$
  does not remove any directed triangles
  containing arcs of any other~$H'\in\packing$
  with $H'\neq H$.
  By \cref{fast-k-small-gen}, we can then reject the instance if~$k > (2t+1)\ell$.  Otherwise, we can find a solution in~$\Time(G,(2t+1)\ell)$ time.

  \eqref{fvs-kern}
  First, we apply \cref{fast-general-rule}
  once for each~$H\in\packing$ in
  $O(|\packing|n^3 + \sum_{H\in\packing}\Time(H,t))$~time.
  After one pass of this rule,
  the instance is reduced
  since reversing arcs in some~$H\in \packing$
  does not remove any directed triangles
  containing arcs of any other~$H'\in\packing$
  with $H'\neq H$.
  Afterwards, by \cref{fast-k-small-gen}, we can reject if~$k > (2t+1)\ell$.  Otherwise, we apply the kernelization algorithm for \textsc{FAST} by~\citet{PPT16} to the instance~$(G,k)$ to obtain an equivalent instance~$(G',k')$ with at most~$4k\le (8t+4)\ell$ vertices and a solution size~$k'\le k$. 
  Hence,
  $(G',\emptyset,k')$ is our problem kernel 
  with parameter~$\ell'=k'\le(2t+1)\ell$ of \FASTGP{}.  \qed\end{proof}

\noindent In \cref{fast-general-fpt}, we again assume that $\Time$~is monotonically nondecreasing in both the size of~$G$ and in~$k$.  As mentioned earlier, $2^{O(\sqrt{k})}+ |V(G)|^{O(1)}$ is the currently best known running time for $\Time(G,k)$~\cite{KS10}.

\begin{corollary}\label[corollary]{cor:fast-abv-cost-t}
  \FASTGP{}
  \begin{enumerate}[(i)]
  \item\label{fast-searchtree} can be solved in~$2^{O(\sqrt{(2t+1)\ell})} + n^{O(1)} + n2^{O(\sqrt t)}$~time, and
  \item\label{fast-kern} admits a problem kernel with at most~$(8t+4) \ell$ vertices computable in~$n^{O(1)} + n2^{O(\sqrt t)}$~time.
  \end{enumerate}
\end{corollary}

\section{Cluster Editing}
\label{sec:cluster-edit}
We finally apply our framework from \cref{sec:approach} to \textsc{Cluster Editing}, a well-studied edge modification problem in parameterized complexity~\cite{Boeck12,CM12,FKP+11,KU12}. %
\decprob{\textsc{Cluster Editing}}{A graph~$G=(V,E)$ and a natural number~$k$.}  {Is there an edge modification set~$S\subseteq \binom{V}{2}$ of size at most~$k$ such that~$G\symdiff{} S$ is a cluster graph, that is, a disjoint union of cliques?}  A graph is a cluster graph if and only if it is~$P_3$-free~\cite{SST04}. Thus, \textsc{Cluster Editing} is the problem of destroying all~$P_3$s by few edge modifications. For brevity, we refer to the connected components of a cluster graph (which are cliques) and to their vertex sets as \emph{clusters}. %
The currently fastest algorithm for
\textsc{Cluster Editing} parameterized by the solution size~$k$ runs in~$O(1.62^k + n+m)$ time~\cite{Boeck12}.
Assuming the
exponential-time hypothesis, \textsc{Cluster Editing} cannot be solved
in~$2^{o(k)}\cdot n^{O(1)}$ time~\cite{FKP+11,KU12}. \textsc{Cluster
  Editing} admits a problem kernel with at most~$2k$
vertices~\cite{CM12}.%

First, in \cref{sec:cet}, we present a fixed-parameter algorithm and problem kernel for \textsc{Cluster Editing} parameterized above lower bounds given by cost-\(t\) packings.  Then, in \cref{sec:p3p}, we present a faster fixed-parameter algorithm for lower bounds given by vertex-disjoint packings of~\(P_3\)s.

\subsection{A fixed-parameter algorithm for vertex-disjoint cost-\(t\) packings}
\label{sec:cet}
Several kernelizations for \textsc{Cluster Editing}
are based on the following observation:
If $G$~contains a clique
such that all vertices in this clique
have the same closed neighborhood,
then there is an optimal solution that
puts these vertices into the same cluster
\cite{PSS09,Guo09,CM12}.
This implies that the edges of this clique
are never deleted.
The following rule
is based on a generalization of this observation.
\begin{rrule}
  \label[rrule]{rule:opt-graph-ce}
  If $G=(V,E)$~contains an induced subgraph~$H=(W,F)$
  having an optimal solution~$S$ of size~$\tau(H)$
  such that,
  for all vertices~$u,v\in W$,
  \begin{itemize}
  \item
    $N_G(v)\setminus W=N_G(u)\setminus W$
    if~$u$ and~$v$ are in the same cluster
    of~$H\symdiff{} S$,
    and
  \item
    $N_G(v)\cap N_G(u)\subseteq W$
    otherwise,
  \end{itemize}
  then replace~$G$ by~$G\symdiff{} S$, remove~$H$ from~$\packing$, and decrease~$k$ by~$\tau(H)$.
\end{rrule} 
An example of \cref{rule:opt-graph-ce} is presented
in \cref{fig:opt-graph-ce}.
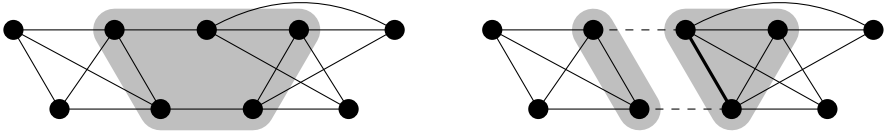
\begin{figure}[t]
  \centering
  \begin{tikzpicture}[x=0.7cm,y=0.7cm]
    \tikzstyle{packing} = [fill,color=lightgray,line cap=round, line
    join=round, line width=16pt]
    \tikzstyle{deleted} = [dashed]
    \tikzstyle{inserted} = [very thick]
    \tikzstyle{vertex} = [color=black,fill=black,circle]
    \begin{scope}[rotate=30]
      \node[vertex] (x1) at (0:1) {};
      \node[vertex] (x2) at (120:1) {};
      \node[vertex] (x3) at (-120:1) {};
    \end{scope}

    \begin{scope}[shift={(-2.6,-0.5)}]
      \begin{scope}[rotate=-30]
        \node[vertex] (z1) at (0:1) {};
        \node[vertex] (z2) at (120:1) {};
      \end{scope}
    \end{scope}

    \begin{scope}[shift={(-4.5,-0.5)}]
      \begin{scope}[rotate=-30]
        \node[vertex] (w1) at (0:1) {};
        \node[vertex] (w2) at (120:1) {};
      \end{scope}
    \end{scope}
  
    \begin{scope}[shift={(1.8,0)}]
       \begin{scope}[rotate=-90]
         \node[vertex] (y1) at (0:1) {};
         \node[vertex] (y2) at (120:1) {};
         \end{scope}
    \end{scope}
    \draw (x1.center)--(x2.center);
    \draw (z1.center)--(z2.center);
    \draw (z1.center)--(x3.center);
    \draw (z2.center)--(x2.center);
    \draw (x3.center)--(x1.center) ;
    \draw (z1.center)--(w1.center);
    \draw (z1.center)--(w2.center);
    \draw (z2.center)--(w1.center);
    \draw (z2.center)--(w2.center);
    \draw (w1.center)--(w2.center);
    \draw (x3.center)--(y1.center);
    \draw (x2.center)--(y1.center);
    \draw (x1.center)--(y1.center);
    \draw (x3.center)--(y2.center);
    \draw (x2.center) to [bend left =30] (y2.center);
    \draw (x1.center)--(y2.center);
    \begin{pgfonlayer}{background}
       \draw[packing] (z2.center)--(x1.center)--(x3.center)--(z1.center)--cycle;
    \end{pgfonlayer}

    \begin{scope}[shift={(9,0)}]
    \begin{scope}[rotate=30]
      \node[vertex] (x1) at (0:1) {};
      \node[vertex] (x2) at (120:1) {};
      \node[vertex] (x3) at (-120:1) {};
    \end{scope}

    \begin{scope}[shift={(-2.6,-0.5)}]
      \begin{scope}[rotate=-30]
        \node[vertex] (z1) at (0:1) {};
        \node[vertex] (z2) at (120:1) {};
      \end{scope}
    \end{scope}

    \begin{scope}[shift={(-4.5,-0.5)}]
      \begin{scope}[rotate=-30]
        \node[vertex] (w1) at (0:1) {};
        \node[vertex] (w2) at (120:1) {};
      \end{scope}
    \end{scope}
  
    \begin{scope}[shift={(1.8,0)}]
       \begin{scope}[rotate=-90]
         \node[vertex] (y1) at (0:1) {};
         \node[vertex] (y2) at (120:1) {};
         \end{scope}
    \end{scope}
    \draw (x1.center)--(x2.center);
    \draw (z1.center)--(z2.center);
    \draw [deleted] (z1.center)--(x3.center);
    \draw [deleted] (z2.center)--(x2.center);
    \draw (x3.center)--(x1.center) ;
    \draw [inserted] (x3.center)--(x2.center) ;
    \draw (z1.center)--(w1.center);
    \draw (z1.center)--(w2.center);
    \draw (z2.center)--(w1.center);
    \draw (z2.center)--(w2.center);
    \draw (w1.center)--(w2.center);
    \draw (x3.center)--(y1.center);
    \draw (x2.center)--(y1.center);
    \draw (x1.center)--(y1.center);
    \draw (x3.center)--(y2.center);
    \draw (x2.center) to [bend left = 30] (y2.center);
    \draw (x1.center)--(y2.center);
    \begin{pgfonlayer}{background}
       \draw[packing] (z2.center)--(z1.center)--cycle;        \draw[packing] (x1.center)--(x2.center)--(x3.center)--cycle; 
    \end{pgfonlayer}

    \end{scope}

  \end{tikzpicture}
  \caption{An illustration of \cref{rule:opt-graph-ce}. Left: An
    induced subgraph~$H$ (highlighted by the gray background)
    fulfilling the conditions of \cref{rule:opt-graph-ce}. Right: The
    result of applying \cref{rule:opt-graph-ce}. The two clusters
    in~$G[W]$ produced by the optimal solution for~$H$ are highlighted
    by a gray background.  }
\label{fig:opt-graph-ce}
\end{figure}
\begin{lemma}\label[lemma]{lem:opt-cluster-correct}
  \cref{rule:opt-graph-ce} is correct.
\end{lemma}

\begin{proof}
  Let \(I':=(G\symdiff S,\packing',k-\tau(H))\) be
  the instance obtained by applying 
  \cref{rule:opt-graph-ce} to~$I:=(G,\packing,k)$
  for some induced subgraph~\(H=(W,F)\) of~\(G=(V,E)\).
  If \(I'\)~is a yes-instance, then so is~\(I\):
  adding the \(\tau(H)\)~edges in~\(S\)
  to any solution~\(S'\) of size~\(k-\tau(H)\)
  for~\(G\symdiff S\)
  gives a solution of size~\(k\)
  for~\(G\)
  since \(S\cap S'=\emptyset\)
  and, thus,
  \(G\symdiff (S\cup S')=(G\symdiff S)\symdiff S'\).
  It remains to prove that
  if \(I\)~is a yes-instance,
  then so is~\(I'\).
  To this end, we show that
  $(G,\packing,k)$~has an optimal solution~$S'$
  such that~$S\subseteq S'$. 

  For convenience,
  let~$X:=V\setminus W$ denote
  the set of vertices not in~$W$.
  Let~$S^*$ be any optimal solution for~$G$,
  and denote by~$G^*:=G\symdiff{} S^*$
  the cluster graph produced by~$S^*$.
  We show how to transform~$S^*$
  into an optimal solution~$S'\supseteq S$.
  To this end,
  partition~$S^*$ as follows:
  \begin{itemize}
  \item $S^*_1:=\{\{u,v\}\in S^*\mid u,v\in X\}$ containing
    all edge modifications outside of~$H$,
  \item $S^*_2:=\{\{u,v\}\in S^*\mid u\in X \wedge v\in W\}$
    containing the edge modifications between~$H$ and the rest of~$G$, and
  \item $S^*_3:=\{\{u,v\}\in S^*\mid u,v\in W\}$.
  \end{itemize}
  Moreover, let~$\Wext$ be the vertices of~$W$
  that have at least one neighbor in~$X$
  and let~$\Wint:=W\setminus \Wext$.
  Consider the following equivalence relation~$\sim$
  on~$\Wext$: two vertices~$u,v\in \Wext$ are equivalent with respect to~\(\sim\) if and only
  if~$N_G(u)\cap X = N_G(v)\cap X$.  For each vertex~$u\in \Wext$, let $[u]$~denote the
  equivalence class of~$u$ in~$\sim$.

  Fix within each equivalence class~$[u]$ of~$\sim$ an
  arbitrary vertex that is incident to a minimum number of edge
  modifications in~$S^*_2$ and, for each vertex~$u$, denote this vertex
  by~$\tilde{u}$.
  Furthermore, for each cluster~$K$ of~$G^*$
  containing some vertices of~$X$ and some vertices
  of~$\Wext$, fix an arbitrary vertex of~$\Wext$ that has in~$G$ a maximum
  number of neighbors in~$K\cap X$; denote this vertex
  by~$u_K$. Finally, call a vertex~$u\in \Wext$~\emph{good} if there
  is a cluster~$K$ such that~$\tilde{u}=u_K$.
  Now consider the edge modification set~$S':=S_1^*\cup\tilde{S}\cup S$, where %
  \begin{align*}
    \tilde{S}:={}&\Bigl\{\{u,v\}\mid \text{$u$ is~good and } \{\tilde{u},v\}\in S_2^*\Bigr\}\cup{}\\
{}\cup{}& \Bigl\{\{u,v\}\mid \text{$u$ is not good and } v\in N_G(u)\cap X\Bigr\}.
  \end{align*}
  Informally, the modifications in~$\tilde{S}$ consider~$\tilde{u}$ to
  determine how to treat all vertices in the equivalence
  class~$[u]$. If, among the vertices of~$\Wext$, $\tilde{u}$ has the most neighbors in its cluster
  in~$G^*$, then all vertices of~$[u]$ are treated
  like~$\tilde{u}$. Otherwise, all edges between vertices in~$[u]$
  and~$X$ are deleted.

  We first show that~$S'$ is a solution, that is, $G':=G\symdiff{} S'$~is a
  cluster graph: First, $G'[X]=G^*[X]$~is a cluster
  graph. Second, $G'[W]=H\symdiff{} S$ and thus it is a cluster
  graph. Third, every vertex~$u \in \Wint$ is contained in a cluster that is a subset of~$\Wint$ in~$G'$: In~$G$, there are no edges between~$\Wint$ and~$X$,~no edges between~$\Wint$ and~$X$ are added by~$\tilde{S}$, and by the first condition on~$S$, no vertex of~$\Wint$ is in the same cluster of~$H\symdiff S$ as a vertex from~$\Wext$. This implies that the connected component of each vertex~$u\in \Wint$ is completely contained in~$W$ and thus it is a clique since~$G'[W]$ is a cluster graph. Finally, consider any equivalence class~$[u]$ of~$\Wext$. By the
  condition of \cref{rule:opt-graph-ce}, $[u]$~is contained in a cluster in~$H\symdiff{} S$.
 Now, if all vertices of~$[u]$ are good, then there is a
  cluster~$K$ in~$G'[X]$ such that in~$G'$ all vertices of~$[u]$ are
  adjacent to all vertices of~$K$.
 Otherwise, no vertex of~$[u]$
  is good and thus, no vertex of~$[u]$ is adjacent to any
  vertex of~$X$ in~$G'$. Finally, if a cluster~$K$ of~$G'[X]$ has neighbors in~$W$, then these edges are only to the vertex
  set~$[u_K]$. Thus, every vertex in~$K$ has neighbors in at most
  one cluster of~$H\symdiff{} S$. Altogether this shows that~$G'$ is a
  cluster graph.

  It remains to show that $|S'|\le |S^*|$.
  First, $S^*_1$~is a subset of~$S^*$ and of~$S'$. Second, $|S|\le |S_3^*|$
  since~$S$ is an optimal solution for~$H$. Thus, it remains
  to show~$|\tilde{S}|\le |S^*_2|$. Since the vertices of~$\Wint$ are not incident to any edge modifications in~$\tilde{S}$, we may prove this inequality by showing, for each
  vertex~$u\in \Wext$, that~$$|\{e\in \tilde{S} \mid u\in e\}|\le |\{e\in S^*_2 \mid
  u\in e\}|.$$
  If~$u$ is good, then the number of edge modifications incident
  with~$u$ in~$\tilde{S}$ is the same as the number of edge
  modifications incident with~$\tilde{u}$ in~$S^*_2$, because~$u$
  and~$\tilde{u}$ have the same neighborhood in~$X$ by the condition of
  the rule. By the choice of~$\tilde{u}$, this is at most as large as the
  number of edge modifications incident with~$u$ in~$S_2^*$ and the
  claim holds.
  
  Otherwise, all edges between~$u$ and~$X$ are deleted
  by~$\tilde{S}$. Let~$K$ denote the cluster in~$G\symdiff{} S^*$
  containing~$\tilde{u}$. Since~$u$ is not good, there is a vertex~$w$
  that is not in~$[u]$ such that~$w$ has at least as many neighbors
  in~$K\cap X$ as~$\tilde{u}$. By the condition of
  \cref{rule:opt-graph-ce} and since~$[\tilde{u}]\neq [w]$, the
  neighborhoods of these two vertices in~$X$ are disjoint, that
  is, $(N_G(\tilde{u})\cap X)\cap (N_G(w)\cap X)=\emptyset$. Since~$w$
  has at least as many neighbors in~$K\cap X$ as~$\tilde{u}$, this means
  that~$$|N_G(\tilde{u})\cap K\cap X|\le |K\cap X|/2.$$ %
  Now, observe that $S^*_2$~contains an edge insertion between~$\tilde{u}$
  and each nonneighbor of~$\tilde{u}$ in~$K\cap X$. Thus, at least~$|K\cap X|/2$ edges between~$\tilde{u}$ and~$K\cap X$ are inserted by~$S^*_2$. Moreover,~$S^*_2$ contains an edge deletion between~$\tilde{u}$ and each neighbor of~$\tilde{u}$ in~$X\setminus K$.
  Altogether, this implies $$|\{e\in S^*_2 : \tilde{u}\in e\}| \ge |K\cap X|/2 + |N_G(\tilde{u})\cap (X\setminus K)| \ge |N_G(\tilde{u})\cap X|.$$
  Therefore, the number of edge modifications incident with~$u$ in~$\tilde{S}$ (this is exactly the number of edges  
  between~$u$ and~$X$) is at most as large as the number of edge modifications incident with~$\tilde{u}$ in~$S^*_2$.
  By the choice of~$\tilde{u}$, this implies the claim.  \qed
\end{proof}
It remains to analyze the running time
for applying \cref{rule:opt-graph-ce}.

\begin{lemma}\label[lemma]{lem:opt-cluster}
  Let $\Time(G,k)$~be the running time used
  for finding an optimal solution of size at most~\(k\)
  in a graph~\(G\) if it exists.

  Then, in $O(m+\sum_{H\in \packing} \Time(H,t))$~time,
  we can apply \cref{rule:opt-graph-ce}
  to all graphs in~$\packing$.
\end{lemma}
Here, we assume that
$\Time$~is monotonically nondecreasing in~$k$
and polynomial in the size of~$G$. Currently,
$O(1.62^k+|V(G)|+|E(G)|)$ is the best known bound
for $\Time(G,k)$~\cite{Boeck12}.
\begin{proof} 
  We first show that the rule can be applied
  in $O(|W|+\sum_{w\in W} \deg_G(w)+ \Time(H,t))$~time
  to an arbitrary graph~$H=(W,F)\in \packing$.
  For convenience, denote~$X:=V\setminus W$.

  First, observe that a necessary condition for the rule is that, for each pair of
  vertices~$u$ and~$v$ in~$H$, their neighborhoods in~$X$ are the same or disjoint.
 This can be checked in
  $O(|W|+\sum_{w\in W}\deg_G(w))$~time as follows.  First, build the bipartite graph with
  parts~$W$ and~$N(W)$ and those edges between the vertex sets that are also edges of~$G$
  (equivalently, $G[W\cup N(W)]$ minus the edges with both endpoints in~$W$ or both
  endpoints in~$N(W)$). This bipartite graph is a disjoint union of complete bipartite
  graphs if and only if above condition is fulfilled. Thus, we check in~$O(|W|+|F|)$ time,
  whether the graph is a disjoint union of complete bipartite graphs. If not, then the
  rule does not apply.
  Otherwise, in the created bipartite graph,
  we compute in $O(|W|+|F|)$~time
  the groups of vertices of~$W$ whose neighborhood is the same.
  Afterwards, we can check in $O(1)$~time
  whether $u$~and~$v$ have the same neighborhood in~$X$ in~\(G\)
  by checking whether they belong to the same group.
  We now compute the set~$S'$ of
  edge modifications that is already determined by the conditions of the rule: If~$v$
  and~$w$ are nonadjacent and have the same nonempty neighborhood in~$X$, then the
  edge~$\{v,w\}$ needs to be inserted and is thus added to~$S'$.  Similarly, if~$v$
  and~$w$ are adjacent and have different neighborhoods in~$X$, then~$\{v,w\}$ needs to be
  deleted and is thus added to~$S'$. Observe that if~$S'$ is a subset of an optimal
  solution~$S$ for~$H$, then~$|S|\le |F|$. Hence, at most~$|F|$ edges are added
  to~$S'$. Since we already computed the groups of vertices that have the same or disjoint
  neighborhoods, we can thus compute~$S'$ in~$O(|W|+|F|)$ time.

  Let~$\Wext$ denote the vertices of~$W$ that have at least one neighbor in~$N(W)$ and
  let~$\Wint:=W\setminus \Wext$. Note that after applying~$S'$, that is,
  in~$H \symdiff{} S'$, the vertices of~$\Wext$ are in separate clusters:
  $(H\symdiff{} S')[\Wext]$~is a cluster graph and
  there are no edges between~$\Wext$ and~$\Wint$ in~$H\symdiff{} S'$.
  Thus, to determine whether~$S'$ can be extended
  to an optimal solution~$S$ for~$H$
  that fulfills the condition of the rule,
  we compute an optimal solution of~$H[\Wint]$ in
  $\Time(H[\Wint],\tau(H)-|S'|)$~time.
  Since $|H[\Wint]|\le |H|$ and $\tau(H)-|S'|<t$, this can be done in
  $\Time(H,t)$~time. The size of the resulting solution~$S$ is compared with the size of
  an optimal solution of~$H$, which can also be computed in $\Time(H,t)$~time.

  It remains to show that \cref{rule:opt-graph-ce}
  can be applied to all graphs
  within the claimed running time.
  For each graph~$H=(W,F)\in \packing$,
  we can check
  in $O(|W|+\sum_{w\in W} \deg_G(w)+\Time(H,t))$~time
  whether \cref{rule:opt-graph-ce} applies.
  If yes,  then we can apply the rule in $O(|F|)$~time
  by modifying at most $|F|$~edges.
  Summing up over all graphs in~$\packing$
  gives the claimed running time.
\qed\end{proof}

\noindent Observe that since~$k$ is decreased by~$\tau(H)$, the parameter~$\ell$ does not increase when \cref{rule:opt-graph-ce} is applied. As for the previous problems, applying the rule to each~$H\in \packing$ is sufficient for bounding~\(k\) in~\(t\) and~\(\ell\) and thus, for transferring known fixed-parameter tractability results for the parameter~\(k\) to the combined parameter~\((t,\ell)\).

\begin{lemma}\label[lemma]{lem:ce-small-packing}
  Let~$(G,\packing,k)$ be a yes-instance of~\CEGP{} such that
  \cref{rule:opt-graph-ce} does not apply to any~$H\in
  \packing$.  Then, $k\le (2t+1)\ell$.
\end{lemma}
\begin{proof}
  Since the instance is reduced
  with respect to \cref{rule:opt-graph-ce},
  for each~$H=(W,F)$ in~$\packing$ and
  each size-$\tau(H)$ solution~$S$ for~$H$,
  either the vertices of some cluster
  have different neighborhoods in~$V\setminus W$
  or two vertices of two distinct clusters
  have a common neighbor outside of~$W$. 

  Now,
  fix an arbitrary optimal solution~$S$ for~$G$.
  By the observation above,
  there are the following two possibilities
  for how $S$~modifies each~$H\in \packing$:
  \begin{enumerate}[(a)]
  \item\label{ce-inside} more than~$\tau(H)$ vertex pairs of~$H$ are modified by~$S$, or
  \item\label{ce-outside} at least one external vertex pair for~$H$ is modified.
  \end{enumerate}
Therefore, $S$ fulfills the condition of \cref{lem:gen-bound} and thus $k\le (2t+1)\ell$.
\qed\end{proof}

\begin{theorem}\label{thm:ce-time}
Let $\Time(G,k)$~be the running time used for finding an optimal solution of size at most~\(k\) in a graph~\(G\) if it exists. Then,  \CEGP{}
\begin{enumerate}[(i)]
\item\label{ce1} is solvable in~$O(\Time(G,(2t+1)\ell) +
  nm+\sum_{H\in \packing}\Time(H,t))$ time and
\item\label{ce2} admits a problem kernel with at most~$(4t+2)\ell$ vertices, which can be computed in~$O(nm+\sum_{H\in \packing}\Time(H,t))$ time.
\end{enumerate}
\end{theorem}
\begin{proof}
  \eqref{ce2}
  First, apply \cref{rule:opt-graph-ce} exhaustively
  in~$O(m+\sum_{H\in \packing}\Time(H,t))$ time.
  Then,
  by \cref{lem:ce-small-packing},
  we can either return ``no''
  or have~$k\le (2t+1)\ell$.
  In the latter case,
  we apply a kernelization algorithm
  for \textsc{Cluster Editing} to the instance~$(G,k)$
  (that is, without~$\packing$),
  which produces, in $O(nm)$~time,
  a problem kernel~$(G',k')$
  with at most~$2k\le (4t+2)\ell$ vertices
  and with~$k'\le k$~\cite{CM12}.
  Adding an empty packing gives
  an equivalent instance~$(G',\emptyset,k')$
  with parameter~$\ell'=k'$ of \CEGP.

  \eqref{ce1} First, apply the kernelization.
  Then, by \cref{lem:ce-small-packing},
  we can either return ``no''
  or have $k\le (2t+1)\ell$.
  We can now apply the algorithm
  for \textsc{Cluster Editing}
  that runs in $\Time(G,(2t+1)\ell)$~time. \qed\end{proof} 
By plugging in the best known bound for~$\Time(G,k)$,
we obtain the following.
\begin{corollary}\label[corollary]{cor:ce-abv-cost-t}
\CEGP{} 
\begin{enumerate}[(i)]
  \item\label{cor-td-searchtree} can be solved in~$O(1.62^{(2t+1)\cdot \ell} + nm + n\cdot 1.62^t)$~time, and
  \item\label{cor-td-kern} admits a problem kernel with at most~$(4t+2)\ell$ vertices that can be computed in~$O(nm+n\cdot 1.62^t)$ time.
  \end{enumerate}
\end{corollary}

\subsection{A Search Tree Algorithm for $P_3$-Packings}
\label{sec:p3p}
For \CEPP,
the generic algorithm based on \cref{rule:opt-graph-ce}
(with $t=1$)
using the currently best running time
for \textsc{Cluster Editing}
leads to a running time of~$O(4.26^\ell+n+m)$.
We now show an algorithm that runs
in \(O(4^\ell\cdot \ell^3 + n + m)\)~time.
The algorithm is based on two special cases
of \cref{rule:opt-graph-ce},
one further reduction rule
and a corresponding branching algorithm.

The analysis of the search tree size is done by considering branching vectors and their corresponding branching number.  The components of a branching vector denote the decrease of the parameter in each recursive branch.  The branching number depends only on the branching vector, the largest branching number gives the base in the upper bound on the search tree size; for further details refer to the relevant monographs~\cite{FK10,Nie06}.

We use the following special cases of \cref{rule:opt-graph-ce}. In both cases, $H$~is a~$P_3$; the correctness is directly implied by \cref{lem:opt-cluster}. In the first rule, adding an edge is a solution which fulfills the condition of \cref{rule:opt-graph-ce}. For convenience, we denote by~$uvw$ a~$P_3$ on the vertices~$u$, $v$, and~$w$, where $v$~is the degree-two vertex.
\begin{rrule}\label[rrule]{rule:almost-twin}
  If~$G$ contains a~$P_3$~$uvw$ such that $N(u)\setminus
  \{u,v,w\}=N(v)\setminus \{u,v,w\} =N(w)\setminus \{u,v,w\}$, then
  insert~$\{u,w\}$ and decrease~$k$ by one.
\end{rrule}
In the second rule, deleting an edge gives such a solution.
\begin{rrule}\label[rrule]{rule:almost-twin2}
  If~$G$ contains a~$P_3$~$uvw$ such that~$N(u)\setminus
  \{u,v,w\}=N(v)\setminus \{u,v,w\}$ and~$N(u) \cap N(w)=\{v\}$, then
  delete~$\{v,w\}$ and decrease~$k$ by one.
\end{rrule}
The third rule is crucial for showing an improved running time.
\begin{rrule}
  \label[rrule]{rule:clique-with-ends}
  If~$G$ contains a clique~$K$
  on at least three vertices such that
  \begin{itemize}
  \item
    every vertex in~$K$ has
    at most one neighbor in~$N(K)$ and
  \item
    every vertex in~$N(K)$ has
    exactly one neighbor in~$K$,
  \end{itemize}
  then delete all edges between~$K$ and~$N(K)$
  and decrement~$k$ by~$q:=|N(K)|$.
\end{rrule}

\begin{lemma}
  \label[lemma]{lem:clique-rule-correct}
  \cref{rule:clique-with-ends} is correct
  and can be exhaustively applied
  in $O(nm)$~time.
\end{lemma}
\begin{proof}
  First, enumerate the set of all maximal cliques~$K$
  on at least three vertices
  such that every vertex of~$K$
  has at most one neighbor in~$N(K)$.
  This can be done in $O(m)$~time \cite{KHMN09}.
  Note that every vertex is contained
  in at most one such clique.
  Now, by scanning
  through the adjacency lists of all vertices
  in an enumerated clique,
  we can identify a vertex that has
  more than one neighbor in the clique. If there is such a vertex, then the clique can be discarded.
  Otherwise, the clique fulfills the conditions
  of the rule and the rule can be applied.
  Thus, one application of the rule takes~$O(m)$ time.
  Since the rule decreases the number of vertices in~$G$,
  it can be applied~$O(n)$ times.

  Since~$q\le |K|$ and~$|K|\ge 3$,
  one can construct $q$~$P_3$s,
  each containing two vertices from~$K$
  and one vertex from~$N(K)$,
  such that no two of them share more than one vertex.
  Thus, at least $q$~edge modifications are needed
  to destroy all~$P_3$s
  that contain at least one vertex~$v\in K$.
  Deleting all $q$~edges between~$K$ and~$N(K)$
  destroys all $P_3$s that
  contain at least one vertex of~$K$.
  Moreover, since these edge deletions cut~$K$
  from the rest of the graph,
  one can safely combine
  any optimal solution for~$G[V\setminus K]$
  with these~$q$ edge deletions,
  which are necessary and sufficient
  to destroy all $P_3$s that contain
  at least one vertex of~$K$,
  to obtain an optimal solution
  that deletes all $q$~edges between~$K$ and~$N(K)$.
\qed\end{proof}
The final rule simply removes isolated clusters from~$G$. 
\begin{rrule}\label[rrule]{rule:isolated-clusters}  If~$G$ contains a connected component~$K$ that is a clique, then remove~$K$ from~$G$.
\end{rrule}

\noindent We can now show our improved algorithm for \CEPP{}.
\begin{theorem}
  \label{thm:ce-abv}
  \CEPP{} can be solved
  in $O(4^{\ell}\cdot \ell^3 + m + n)$~time.
\end{theorem}
\begin{proof}
  We prove the theorem using a branching algorithm,
  which applies several branching rules.
  Herein, we assume
  that $\packing$~contains at least one~$P_3$~$uvw$.
  Otherwise,
  the graph is either~$P_3$-free
  (in this case we are done)
  or we can add a~$P_3$ to~$\packing$,
  which increases~$|\packing|$ by one and thus
  reduces the parameter.
  Furthermore, we assume
  that \cref{rule:almost-twin,rule:almost-twin2,%
    rule:clique-with-ends,rule:isolated-clusters}
  do not apply.
  First, we take care of~$P_3$s
  that do not share an edge with a packing~$P_3$.
  
  \emph{Branching Rule~1: If there is an induced~$P_3$
    that contains at most one vertex
    of each $P_3$ in~$\packing$, then
  branch into the three cases
  to destroy this~$P_3$.}

  None of the cases destroys a~$P_3$ of~$\packing$.
  Thus, the parameter is decreased by one in each case;
  the branching number is 3.

  The three further rules deal with packing~$P_3$s $uvw$.

  \emph{Branching Rule~2: If there is a vertex~$x$ that is adjacent to~$u$ and~$w$ but not
    to~$v$, then branch into four cases: delete~$\{u,x\}$; delete~$\{w,x\}$;
    add~$\{v,x\}$; or delete~$\{u,v\}$ and~$\{v,w\}$ and add~$\{u,w\}$.}
  
  In each of the first three branches, $k$~is reduced by one without
  destroying any~$P_3$ of~$\packing$. If none of the first three cases
  applies, then~$u$, $w$, and~$x$ are in the same cluster (no edge
  deletions between~$u$ or~$w$ and~$x$) and~$v$ is not in this
  cluster.
  This makes the three edge modifications
  in~$G[\{u,v,w\}]$ necessary.
  Thus, $k$~is reduced by three and~$|\packing|$ is
  reduced by one in this case.
  The resulting branching vector is~$(1,1,1,2)$,
  which gives the branching number~$3.31$.

  \emph{Branching Rule~3: If there are vertices~$x$ and~$y$ such that
    \begin{itemize}
    \item $x$ is adjacent to~$u$ and~$v$, and
    \item $y$ is adjacent to exactly one vertex
      of~$\{u,v\}$, 
    \end{itemize}
    then branch into four cases: delete~$\{u,x\}$; delete~$\{v,x\}$; delete the edge
    between~$y$ and its neighbor in~$\{u,v\}$; or add the edge between~$y$ and its
    nonneighbor in~$\{u,v\}$.
  }

  In each case, an edge is modified without destroying
  any~$P_3$ of~$\packing$. If none of the first three cases applies,
  then~$u,v,x$, and $y$ are in the same cluster,
  which means that the missing edge
  between~$y$ and either~$u$ or~$v$
  has to be added.
  Since the parameter is reduced by one in each branch, the branching
  number is~$4$.
 
  \emph{Branching Rule~4: If there is a vertex~$x$ that is adjacent to~$v$ and not
    adjacent to~$u$ and~$w$, then branch into four cases:
  delete~$\{v,x\}$; add~$\{u,x\}$; add~$\{w,x\}$;
  or delete~$\{u,v\}$ and~$\{u,w\}$.}

  If none of the first three cases applies,
  then any cluster containing~$v$
  contains neither~$u$ nor~$w$,
  thus the branching is correct.
  The parameter is reduced by one in each
  branch, as the last branch destroys a $P_3$ of~$\packing$ but
  reduces~$k$ by two. Thus, the branching number is~$4$.

  These are the only branching rules that are performed. We now show, by a case
  distinction, the following: If none of the branching rules and reduction rules applies,
  then the remaining graph has maximum degree two and we can solve the problem in
  polynomial time.

  \emph{Case I: $\packing$ contains a~$P_3$ $uvw$ such that~$u$ and~$w$ have a common
    neighbor~$x\neq v$.} Since Branching~Rule~2 does not apply, $x$ is also a
  neighbor of~$v$. Since Branching~Rule~3 does not apply, we have that every other vertex~$y$ that is adjacent to~$u$ is also adjacent to~$v$ and vice versa. Similarly, every other vertex~$y$ that is adjacent to~$w$ is also adjacent to~$v$. Thus,~$u$, $v$,
  and~$w$ have the same neighbors in~$V\setminus \{u,v,w\}$. 
This contradicts our assumption that
\cref{rule:almost-twin} does not apply.

  \emph{Case II: $\packing$ contains a~$P_3$ $uvw$ such that
    $v$ has degree at least three.}
  Since Branching~Rule~4 does not apply,
  each vertex~$x\in N(v)\setminus \{u,w\}$ is a neighbor
  of~$u$ or~$w$ and,
  since Case~I does not apply,
  it is not a neighbor of both.
  Moreover, since Branching~Rule~3 does not apply,
  $v$~can have common neighbors with at most one
  of~$u$ and~$w$.
  Thus, without
  loss of generality, $u$ and~$v$ have the same neighborhood
  in~$V\setminus \{u,v,w\}$ and~$v$ and~$w$ have no common neighbors.
  This contradicts our assumption that 
  \cref{rule:almost-twin2} does not apply.

  \emph{Case III: $\packing$ contains a~$P_3$ $uvw$ such that~$u$ has degree at least three.}  
  Consider~$p$ and~$q$ from~$N(u)\setminus \{v\}$.
  Since Case~II does not apply, $p$ and~$q$ are not middle vertices
  of a $P_3$ in~$\packing$.
  Moreover, since Case~I does not apply, 
  $p$ and~$q$ are not from the same~$P_3$ of~$\packing$.
  Consequently, if $G[\{u,p,q\}]$~is a~$P_3$, then Branching Rule~1 applies.
  This implies that $G[N[u]\setminus \{v\}]$~is a clique~$K$ of size at least three.
  We now show the following claim, which contradicts our assumption
  that \cref{rule:clique-with-ends} does not apply.
  \begin{quote}
    \emph{Claim:} Each vertex of~$K$ has at most one neighbor in~$V\setminus K$ and every
    vertex in~$N(K)$ has at most one neighbor in~$K$.
  \end{quote}
  First, observe that no vertices from~$K$ are middle vertices of a~$P_3$
  in~$\packing$. Thus, $K$~contains at most one vertex~$x$ from each~$P_3$ of~$\packing$
  and this vertex~$x$ is not a middle vertex. For each such~$x\in K$ from a packing~$P_3$
  $xyz$, the same conditions apply as to~$u$, thus~$G[N[x]\setminus \{y\}]$ is a
  clique~$K'$. This implies~$K=K'$: Since $G[N[u]\setminus \{v\}]$ is a
  clique, $x$~is adjacent to every vertex in~$N[u]\setminus \{v\}$ and since
  $G[N[x]\setminus \{y\}]$ is a clique,~$u$~is adjacent to every vertex in $N[x]\setminus
  \{y\}$. Summarizing, each vertex from~$K$ that is in a~$P_3$ of $\packing$ has exactly
  one neighbor outside of~$K$, this neighbor is a middle vertex of the packing~$P_3$
  containing~$x$. These middle vertices have only one neighbor in~$K$ since they have
  degree two and~$K$ contains only one vertex from each packing~$P_3$.

  Now let~$x$ denote a vertex of~$K$ that is not contained in any~$P_3$ of~$\packing$. We
  show that~$N[x]=K$, which implies the claim. Since the middle vertices of each~$P_3$
  of~$\packing$ have no neighbors outside of this~$P_3$ and since Branching~Rule~2 does not apply,
  we have that~$x$ is adjacent to at most one vertex of each packing~$P_3$. Thus,
  $G[N[x]]$ is a clique~$K'$ as otherwise, Branching~Rule~1 applies. Again, $K'=K$ as
  $G[N[u]\setminus \{v\}]$ being a clique implies that~$x$ is adjacent to every vertex
  in~$N[u]\setminus \{v\}$ and $G[N[x]]$ being a clique implies that~$u$ is adjacent to
  every vertex in $N[x]$.

  \emph{Case IV: otherwise.}
  We show that every vertex has degree at most two.
  This is true for the middle vertices of~$P_3$s in~$\packing$
  as Case~II does not apply.
  This also holds for the endpoints of~$P_3$s in~$\packing$
  as Case~III does not apply.
  We now argue that every other vertex~$x$ cannot have two neighbors. 

The vertex~$x$ has at least one
  neighbor~$u$ from some~$P_3$ $uvw$ of~$\packing$: $x$~is contained in at least one~$P_3$
  because the instance is reduced with respect to \cref{rule:isolated-clusters} and
  this~$P_3$ contains at least two vertices of some~$P_3$ of the packing because Branching~Rule~1
  does not apply.  If~$x$ has a further
  neighbor~$y$, then~$y\neq v$ (since Case II does not apply) and~$v\neq w$ since (Case~I) does not apply.
  Consequently, $u$, $x$, and~$y$ form a triangle (since Branching~Rule~1 does not
  apply). 
  Thus, $u$~has two neighbors outside of his~$P_3$ of~$\packing$,
  which means that Case~III applies.

  Thus, $G$~has maximum degree two and does not contain isolated triangles because the instance is reduced with respect to \cref{rule:isolated-clusters}. In this case, an optimal solution
  can be obtained by computing a maximum matching~$M$ and then
  deleting all edges of~$G$ that are not in~$M$.
 
  Altogether,
  the above considerations imply
  a search tree algorithm
  with search tree size~$O(4^{\ell})$. 
  After an initial kernelization, which,
  due to \cref{cor:ce-abv-cost-t}, runs in $O(m+n)$~time
  for~$t=1$, the instance has~$O(\ell)$ vertices.
  Thus, the steps at each search tree node
  including the reduction rules
  can be performed in~$O(\ell^3)$~time.
  \qed
\end{proof}

\section{Hardness Results for Edge Deletion and Vertex Deletion Problems}
\label{sec:hardness}
In this section, we show edge modification problems and vertex deletion problems that are NP-hard even for small forbidden induced subgraphs and if~$\ell=k-|\packing|$, where \(\packing\) is a vertex-disjoint packing of forbidden induced subgraphs.  We also show that algorithms for \textsc{Vertex Cover} parameterized above lower bounds do not generalize to \(d\)-\textsc{Hitting Set}---the natural generalization of \textsc{Vertex Cover} to hypergraphs.
\subsection{Hard edge deletion problems}
\begin{theorem}\label{k6 hard}
For every fixed $q\ge 6$, \gfdev{K_q} is NP-hard for~$\ell=0$.
\end{theorem}
\noindent We prove \cref{k6 hard} by giving a %
reduction from \POITS{} (\cref{prob:poits}).

\begin{construction}\label[construction]{k6 reduction}
  Let $\phi$ be a Boolean formula
  with variables~$x_1,\ldots,x_n$
  and clauses~$C_1,\allowbreak\ldots,C_m$.
  We assume that each clause~$C_j$ contains exactly three pairwise distinct variables.
  We create a graph~$G$ and a vertex-disjoint $K_q$-packing~$\packing$ as follows.
  
  For each variable~$x_i$, add a clique $X_i$ on~$q$ vertices to~$G$ that has two distinguished disjoint edges~$x_i^\text{F}$ and~$x_i^\text{T}$.  For each clause~$C_j=(l_1\wedge l_2\wedge l_3)$ with literals $l_1,l_2$, and $l_3$, add a clique $Y_j$ on~$q$ vertices to~$G$ that has three distinguished and pairwise disjoint edges~$e_{l_1},e_{l_2}$, and~$e_{l_3}$ (which exist since~$q \ge 6$).  Finally, if $l_t=x_i$, then identify the edge~$e_{l_t}$ with~$x_i^\text{T}$ and if $l_t=\neg x_i$, then identify the edge~$e_{l_t}$ with~$x_i^\text{F}$.
  The packing~$\packing$ consists of all~$X_i$ introduced for the variables~$x_i$ of~$\phi$.
\end{construction}

\begin{lemma}\label[lemma]{only generated cliques}
  Let $G$~be the graph output by \cref{k6 reduction} and let $H$~be
  an induced $K_q$ in~$G$. Then, $H$~is either one of the~$X_i$ or one of the~$Y_j$.
\end{lemma}

\begin{proof}
  First, note that the~$X_i$ are pairwise vertex-disjoint since \cref{k6 reduction} only identifies edges of~$Y_j$s with edges of~$X_i$s and no edge in any~$Y_j$ is identified with edges in different~$X_i$.
  For any~$X_i$ and~$Y_j$, the vertices in~$V(X_i)\setminus V(Y_j)$
  are nonadjacent to those in~$V(Y_j)\setminus V(X_i)$.  Similarly,
  for $Y_i$ and~$Y_j$, the vertices in~$V(Y_i)\setminus V(Y_j)$ are
  nonadjacent to those in~$V(Y_j)\setminus V(Y_i)$ for~$i\neq j$.  Thus, every
  clique in~$G$ is entirely contained in one of the~$X_i$ or~$Y_j$.%
\qed\end{proof}

\noindent \cref{only generated cliques} allows us to prove \cref{k6 hard}.

\begin{proof}[of \cref{k6 hard}]
  We show that $\phi$ is satisfiable if and only if~$G$ can be made $K_q$-free by~$k=|\packing|$~edge deletions (that is, $\ell=0$).

  First, assume that there is an assignment that satisfies~$\phi$.  We construct a \gfdeset{K_q}~$S$ for~$G$ as follows: if the variable~$x_i$ is set to true, then put~$x_i^\text{T}$ into~$S$.  If the variable~$x_i$ is set to false, then add~$x_i^\text{F}$ to~$S$.  Thus, for each~$X_i$, we add exactly one edge to~$S$.  Since~$\packing$~consists of the~$X_i$, we have $|S|=|\packing|$.  Moreover, since each clause~$C_j$ contains a true literal, at least one edge of each~$Y_j$ is contained in~$S$.  Thus, $G\setminus S$ is $K_q$-free, since, by \cref{only generated cliques}, the only $K_q$s in~$G$ are the~$X_i$ and~$Y_j$ and, for each of them, $S$~contains at least one edge.

  Now, assume that $G$~can be made $K_q$-free by deleting a set~$S$ of $|\packing|$~edges.  Then, $S$~deletes exactly one edge of each~$X_i$ and at least one edge of each~$Y_j$.
  We can assume without loss of generality that~$S$ contains either the edge~$x_i^\text{T}$ or~$x_i^\text{F}$ for each~$X_i$ since deleting one of these edges instead of another edge in~$X_i$ always yields a solution by \cref{k6 reduction}.
  Thus, the deletion set~$S$ corresponds to a satisfying assignment for~$\phi$.
\qed\end{proof}

\subsection{Hard vertex deletion problems}
\label{sec:vertex-del}
In this section,
we show NP-hardness of the problem
of destroying all induced paths~\(P_q\) on $q\ge 3$~vertices
by at most~$|\packing|$ vertex deletions
if a packing~\(\packing\) vertex-disjoint induced~$P_q$s
in the input graph~$G$ is provided as input.

\decprob{\PLVD}%
{A graph~$G=(V,E)$, a vertex-disjoint packing~$\packing$ of induced \(P_q\)s, and a natural number~$k$.}%
{Is there a vertex set~\(S\subseteq V\) of size at most~$k$ such that~$G[V\setminus S]$ does not contain~\(P_q\) as induced subgraph?}

\begin{theorem}\label{thm:pl-hard}
 For every fixed~$q\ge 3$, \PLVD{} is NP-hard even if~$\ell=0$.
\end{theorem}

\noindent The reduction is from $q$\textsc{-SAT}:

\begin{construction}\label[construction]{Pl construction}
  Let $\phi$ be a Boolean formula
  with variables $x_1, \ldots , x_n$
  and clauses~$C_1,\allowbreak\ldots,C_m$.
  We assume that each clause~$C_j$
  contains exactly $q$~pairwise distinct variables.
  We construct a graph~$G$
  and a maximal vertex-disjoint packing~$\packing$ of~$P_q$s
  as follows; an illustration of the construction is given in Figure~\ref{fig:pq-cons}.

\begin{figure}[t]\centering
  
  \begin{tikzpicture}[x=0.9cm,y=0.9cm]
    \tikzstyle{packing} = [color=lightgray,line cap=round, line
    join=round, line width=15pt]
    \tikzstyle{deleted} = [dashed]
    \tikzstyle{inserted} = [very thick]
    \tikzstyle{vertex} = [color=black,fill=black,circle]
    \begin{scope}[rotate=30]
      \node[vertex,label=left:$v_{i}^{j-1,T}$] (x1) at (-50:3) {};
      \node[vertex,label=left:$v_{i}^{j-1,F}$] (x2) at (-30:3) {};
      \node[vertex,label=left:$v_{i}^{j,T}$] (x3) at (-10:3) {};
      \node[vertex,label=below left:$v_{i}^{j,F}$] (x4) at (10:3) {};
      \node[vertex,label=below left:$v_{i}^{j+1,T}$] (x5) at (30:3) {};
      \node[vertex,label=left :$v_{i}^{j+1,F}$] (x6) at (50:3) {};
      \node[vertex] (a11) at (-50:4) {};
      \node[vertex] (a21) at (-30:4) {};
      \node[vertex] (a41) at (10:4) {};
      \node[vertex] (a51) at (30:4) {};
      \node[vertex] (a61) at (50:4) {};
    \end{scope}

    \begin{scope}[shift={(10,0)}, rotate=150]
      \node[vertex] (y1) at (-50:3) {};
      \node[vertex] (y2) at (-30:3) {};
      \node[vertex,label=below right:$v_{r}^{j',T}$] (y3) at (-10:3) {};
      \node[vertex] (y4) at (10:3) {};
      \node[vertex] (y5) at (30:3) {};
      \node[vertex] (y6) at (50:3) {};
      \node[vertex] (b11) at (-50:4) {};
      \node[vertex] (b21) at (-30:4) {};
      \node[vertex] (b41) at (10:4) {};
      \node[vertex] (b51) at (30:4) {};
      \node[vertex] (b61) at (50:4) {};
    \end{scope}

    \begin{scope}[shift={(5,8)}, rotate=-90]
      \node[vertex] (z1) at (-50:3) {};
      \node[vertex] (z2) at (-30:3) {};
      \node[vertex] (z3) at (-10:3) {};
      \node[vertex,label=below right:$v_{s}^{j'',F}$] (z4) at (10:3) {};
      \node[vertex] (z5) at (30:3) {};
      \node[vertex] (z6) at (50:3) {};
      \node[vertex] (c11) at (-50:4) {};
      \node[vertex] (c21) at (-30:4) {};
      \node[vertex] (c31) at (-10:4) {};
      \node[vertex] (c51) at (30:4) {};
      \node[vertex] (c61) at (50:4) {};
    \end{scope}

    \foreach \s[evaluate=\s as \sp using int(\s+1)] in {1,...,5}
    {      
      \draw (x\s)--(x\sp.center);
      \draw (y\s)--(y\sp.center);
      \draw (z\s)--(z\sp.center);
    }
    
    \foreach \s in {1,2,4,5,6}
      {      
        \draw (x\s)--(a\s1);
        \draw (y\s)--(b\s1);
      }

      \foreach \s in {1,2,3,5,6}
      {      
        \draw (z\s)--(c\s1);
      }

    \draw (x3) edge[very thick,out=40,in=170] (y3); 
    \draw (y3) edge[very thick,out=160,in=280] (z4);

    \begin{pgfonlayer}{background}
        \draw[packing] (a11.center)--(x1.center)--(x2.center);
        \draw[packing] (a41.center)--(x4.center)--(x3.center);
        \draw[packing] (a51.center)--(x5.center)--(x6.center);
        \draw[packing] (b11.center)--(y1.center)--(y2.center);
        \draw[packing] (b41.center)--(y4.center)--(y3.center);
        \draw[packing] (b51.center)--(y5.center)--(y6.center);
        \draw[packing] (c11.center)--(z1.center)--(z2.center);
        \draw[packing] (c31.center)--(z3.center)--(z4.center);
        \draw[packing] (c51.center)--(z5.center)--(z6.center);
    \end{pgfonlayer}

  \end{tikzpicture}

  \caption{An illustration of Construction~\ref{Pl construction} for~$q=3$. The figure shows parts of the variable cycles for three variables~$x_i,x_r,x_s$ that occur in the clause~$C_t=(x_i\vee x_r\vee \neg x_s)$. The packing~$P_3$s are highlighted by a gray background.}
  \label{fig:pq-cons}

\end{figure}
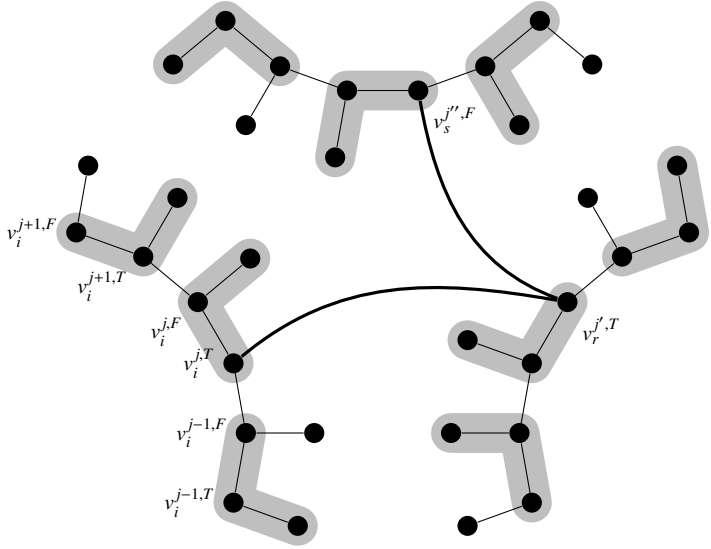

  First,
  we introduce variable gadgets,
  which will ensure that a solution to~\PLVD{}
  corresponds to an assignment of~$\phi$.
  In the following,
  let $\occ(i)$~denote
  the number of occurrences of~$x_i$ and~$\neg x_i$
  in clauses of~$\phi$.
  For each variable~$x_i$, add~$4\occ(i)$ vertices:
  $v_{i}^{j,\text{T}}$ and~$v_{i}^{j,\text{F}}$,
  where~$1\le j\le 2\occ(i)$.
  Call~$v_i^{j,\text{T}}$ a \emph{true vertex}
  and~$v_i^{j,\text{F}}$ a \emph{false vertex}.
  Create an induced cycle on the true and false vertices
  by adding the edge set
  \begin{align*}
    E_i:={}& \Bigl\{\{v_i^{j,\text{T}},v_i^{j,\text{F}}\}
             \mid 1\le j\le 2\occ(i)\Bigr\}
             \cup
             \Bigl\{\{v_i^{j,\text{F}},v_i^{j+1,\text{T}}\}
             \mid 1\le j< 2\occ(i)\Bigr\}
             \cup
             \Bigl\{\{v_i^{2\occ(i),\text{F}},v_i^{1,\text{T}}\}\Bigr\}.
  \end{align*}
  Call this cycle the \emph{variable cycle} of~$x_i$.
  
  Then, for each even~$j$,
  attach to~$v_{i}^{j,\text{T}}$ an induced~\(P_{q-2}\), that is,
  make one of its degree-one vertices adjacent to~$v_{i}^{j,\text{T}}$.
  Then, again for each even~$j$,
  attach to~$v_{i}^{j,\text{F}}$ an induced~\(P_{q-2}\)
  in the same fashion.
  These paths are called
  the \emph{attachment paths} of the $j$th segment
  of the variable cycle of~$x_i$.
  
  Now, for each variable~$x_i$,
  assign to each clause~$C_t$ containing~$x_i$ or~$\neg x_i$
  a unique number~$p \in \{1, \ldots , \occ(i)\}$.
  Consider the number~$j=2p-1$.
  We will use vertex~$v_i^{j,\text{T}}$ or~$v_i^{j,\text{F}}$
  to build the clause gadget for clause~$C_t$.
  If~$C_t$ contains the literal~$x_i$, then attach an
  induced~\(P_{q-2}\) to~$v_{i}^{j,\text{F}}$.
  Otherwise,
  attach an induced~\(P_{q-2}\) to~$v_{i}^{j,\text{T}}$.
  As above,
  call the path the \emph{attachment path}
  of the~$j$th segment of the cycle.
  Now, let~$\num_i^t:=v_i^{j,\text{T}}$
  if $C_t$~contains~$x_i$,
  and let $\num_i^t:=v_i^{j,\text{F}}$
  if $C_t$~contains~$\neg x_i$.
  Call these vertices the \emph{literal vertices} of clause~$C_t$,
  denoted~$\Pi_t$.
  The construction of~$G$ is completed as follows. For each~$\Pi_t$ add an arbitrary set of edges to~$G$ such that~$G[\Pi_t]$ is an induced~$P_q$.
  The $P_q$-packing~$\packing$ contains
  one (arbitrary) attachment path
  plus the two segment vertices
  from each segment
  of each variable cycle.
\end{construction}

\begin{proof}[of \cref{thm:pl-hard}]
  Let $G$ be the graph output by \cref{Pl construction} and let~$\packing$ be
  the $P_q$-packing.
  We show that $\phi$ has a satisfying assignment if and only if $G$ can be made
  $P_q$-free by exactly~$|\packing|$ vertex deletions (that is, $\ell=0$).

  Assume that~$\phi$ has a satisfying assignment.
  For each true variable~$x_i$ in this assignment,
  delete all true vertices in its variable gadget,
  that is, $v_i^{j,\text{T}}$ for $1\le j\le 2\occ(i)$.
  For each false variable~$x_i$ in this assignment,
  delete all false vertices in its variable gadget,
  that is, $v_i^{j,\text{F}}$ for $1\le j\le 2\occ(i)$.
  Denote this vertex set by~$S$ and observe that $|S|=|\packing|$.
  Moreover,
  observe that each vertex on the variable cycle for~$x_i$
  is either deleted
  or both of its neighbors on the cycle are deleted.
  Every~$P_q$ in~$G$
  contains at least one vertex from a variable cycle
  as the attachment paths are too short to induce~$P_q$s. Thus, to show~$P_q$-freeness of~$G[V\setminus S]$ it is sufficient to show that no vertex from a variable cycle is in a~$P_q$.

  Consider an undeleted vertex in the variable cycle for~$x_i$.
  Assume, without loss of generality,
  that this is a true vertex~$v_i^{j,\text{T}}$.
  If~$j$ is even, then~$v_i^{j,\text{T}}$ is not in a~$P_q$
  as its neighbors on the cycle are deleted
  and its only other neighbor is in an attachment path.
  If~$j$ is odd and the
  clause~$C_t$ corresponding to the~$j$th segment of the cycle
  contains~$\neg x_i$, then the only neighbor of~$v_i^{j,\text{T}}$ in~$G[V\setminus S]$
  is in an attachment path. It remains to show that~$v_i^{j,\text{T}}$ is not
  in a~$P_q$ if~$C_t$ contains~$x_i$. The only neighbors
  of~$v_i^{j,\text{T}}$ in~$G[V\setminus S]$ are in~$\Pi_t$.
  Observe that~$G[\Pi_t]$ is an induced~$P_q$
  and that,
  in~$G[V\setminus S]$,
  every vertex on this path is deleted
  or its neighbors in~$V\setminus \Pi_t$ are deleted.
  Hence, the connected
  component of~\(G[V\setminus S]\) containing~$v_i^{j,\text{T}}$
  is an induced subgraph
  of~$G[\Pi_t]$. Since the assignment is satisfying, at least one
  vertex of~$\Pi_t$ is deleted. Thus, this connected component has at
  most~$q-1$ vertices and does not contain a~$P_q$.
 
  Conversely,
  let~$S\subseteq V$ be a size-$|\packing|$ vertex set
  such that~$G[V\setminus S]$ is~$P_q$-free.
  First, observe that,
  without loss of generality,
  for each variable cycle
  either all true or all false vertices are deleted:
  No vertex in an attachment path~$P$ is deleted
  since it is always as good to delete
  the vertex in the variable cycle
  that has a neighbor in~$P$.
  Hence, at least one vertex of each segment is deleted
  since, otherwise,
  one of the~$P_q$'s in~$\packing$ is not destroyed.
  This already requires $|\packing|$~vertex deletions
  and thus \emph{exactly} one vertex for each segment
  of each variable cycle
  is deleted.
  Finally, by construction,
  every adjacent pair of vertices in the variable cycle
  forms a~$P_q$ with some attachment path.
  Therefore, one of the two vertices is deleted,
  which implies that either every even or every odd vertex
  of the cycle is deleted. 

  Hence, the vertex deletions in the variable cycle
  define a truth assignment~$\beta$ to~$x_1,\ldots , x_n$:
  If all true vertices of the variable cycle of~$x_i$ are deleted, 
  then set~$\beta(x_i):=\true$;
  otherwise, set~$\beta(x_i):=\false$.
  This assignment is satisfying:
  Since $G[V\setminus S]$~is $P_q$-free, for each clause~$C_t$,
  at least one vertex $\num_i^t$ of~$\Pi_t$ is deleted.
  Without loss of generality,
  let~$\num_i^t= v_i^{j,\text{T}}$, that is,
  $C_t$~contains the literal~$x_i$.
  Then, $\beta(x_i)=\true$ and thus $\beta$~satisfies clause~$C_t$.
\qed\end{proof}

\noindent \cref{thm:pl-hard} easily transfers to a hardness result for the generalization of \textsc{Vertex Cover} to \(d\)-uniform hypergraphs:

\decprob{\boldmath$d$-\textsc{Uniform Hitting
    Set with Packing}} {A hypergraph~$H=(V,E)$ with $|e|=d$ for
  all~$e\in E$, a set~$\packing\subseteq E$ of pairwise
  vertex-disjoint hyperedges, and an integer~$k$.}{Is there a vertex
  set~$V'\subseteq V$ of size at most~$k$ such that~$\forall e\in E:
  V'\cap e\neq \emptyset$?}

\noindent An instance \((G,\packing,k)\) of \PLVD{} can easily be transformed into an equivalent instance~\((H,\packing,k)\) of \(q\)-\textsc{Uniform Hitting Set with Packing} by taking the hypergraph~\(H\) on the same vertex set as~\(G\) having a hyperedge~\(e\) if and only if \(G[e]\) is a~\(P_q\).  The packing~\(\packing\) and~\(k\) stay unchanged, and so does~\(\ell\).  Thus, we obtain the following result:

\begin{corollary}\label[corollary]{cor:hs}
  For every~$d\ge 3$, $d$-\textsc{Uniform Hitting Set with Packing} is NP-hard even if~$\ell=0$.
\end{corollary}
  
\noindent \cref{cor:hs} shows that the known above-guarantee fixed-parameter algorithms for \textsc{Vertex Cover}~\cite{RO09,CPPW13,LNR+14,GP15} do not generalize to \(d\)-\textsc{Uniform Hitting Set}.

\section{Conclusion}
\label{sec:conclusion}
It is open to extend our framework to further problems. The
most natural candidates appear to be problems where the forbidden induced subgraph has four vertices. Examples are \textsc{Cograph Editing}~\cite{LWGC12} which is
the problem of destroying all induced~$P_4$s, \textsc{$K_4$-free
  Editing}, \textsc{Claw-free Editing}, and \textsc{Diamond-free Deletion}~\cite{FGKNU11,SS15}. Another direction could be to investigate edge completion problems that allow for subexponential-time algorithms~\cite{DFPV15}. In the case of
vertex-deletion problems, \textsc{Triangle Vertex Deletion} appears to
be the most natural open case. Furthermore, it would be nice to obtain
more general theorems separating the tractable from the hard cases for
this parameterization.
For \textsc{Cluster Editing} and \textsc{Triangle Deletion} improved
running times \shorten{by better search trees} are desirable. Maybe more
importantly, it is open to determine the complexity of \textsc{Cluster
  Editing} and \textsc{Feedback Arc Set in Tournaments} parameterized
above the size of \emph{edge-disjoint} packings of forbidden
induced subgraphs.
Finally,  our framework offers an interesting tradeoff between
running time and power of generic data reduction rules. Exploring such
tradeoffs seems to be a rewarding topic for the future. The generic
rules presented in this work can be easily implemented, which asks for
subsequent experiments to evaluate their effectiveness.
\bibliographystyle{spbasic}
\bibliography{ag-editing}

\end{document}